%% file: ms.tex
\documentclass[a4paper,12pt]{article}
\usepackage{amsmath, bm}
\usepackage{amssymb,amsthm,graphicx}
\usepackage{enumitem}
\usepackage{color}
\usepackage{epsfig}
\usepackage{graphics}
\usepackage{pdfpages}
\usepackage{subcaption}
\usepackage[font=small]{caption}
\usepackage[hang,flushmargin]{footmisc} 
\usepackage{float}
\usepackage{rotating,tabularx}
\usepackage{booktabs}
\usepackage[mathscr]{euscript}
\usepackage{natbib}
\usepackage{setspace}
\usepackage{placeins}
\usepackage[left=3cm,right=3cm,bottom=3cm,top=3cm]{geometry}
\numberwithin{equation}{section}
\allowdisplaybreaks[3]

\usepackage{hyperref}
\hypersetup{
hidelinks,
}

\input{macros}

\begin{document}

\heading{Nonparametric comparison}{of epidemic time trends:}{the case of COVID-19}

\authors{Marina Khismatullina\renewcommand{\thefootnote}{1}\footnotemark[1]}{University of Bonn}{Michael Vogt\renewcommand{\thefootnote}{2}\footnotemark[2]}{University of Bonn} 
\footnotetext[1]{Corresponding author. Address: Bonn Graduate School of Economics, University of Bonn, 53113 Bonn, Germany. Email: \texttt{marina.k@uni-bonn.de}.}
\renewcommand{\thefootnote}{2}
\footnotetext[2]{Address: Department of Economics and Hausdorff Center for Mathematics, University of Bonn, 53113 Bonn, Germany. Email: \texttt{michael.vogt@uni-bonn.de}.}
\renewcommand{\thefootnote}{\arabic{footnote}}
\setcounter{footnote}{2}

\vspace{-0.85cm}

\renewcommand{\abstractname}{}
\begin{abstract}
\noindent The COVID-19 pandemic is one of the most pressing issues at present. A question which is particularly important for governments and policy makers is the following: Does the virus spread in the same way in different countries? Or are there significant differences in the development of the epidemic? In this paper, we devise new inference methods that allow to detect differences in the development of the COVID-19 epidemic across countries in a statistically rigorous way. In our empirical study, we use the methods to compare the outbreak patterns of the epidemic in a number of European countries.
\end{abstract}

\renewcommand{\baselinestretch}{1.2}\normalsize

\noindent \textbf{Key words:} simultaneous hypothesis testing; multiscale test; time trend; panel data; COVID-19.

\noindent \textbf{JEL classifications:} C12; C23; C54.


\section{Introduction}

There are many questions surrounding the current COVID-19 pandemic that are not well understood yet. A question which is particularly important for governments and policy makers is the following: How do the outbreak patterns of COVID-19 compare across countries? Are the time trends of daily new infections more or less the same across countries, or is the virus spreading differently in different regions of the world? Identifying differences between countries may help, for instance, to better understand which government policies have been more effective in containing the virus than others. The main aim of this paper is to develop new inference methods that allow to detect differences between time trends of COVID-19 infections in a statistically rigorous way.

Let $\X_{it}$ be the number of new infections on day $t$ in country $i$ and suppose we observe a sample of data $\mathcal{\X}_i = \{ \X_{it}: 1 \le 1 \le T \}$ for $n$ different countries $i$. In order to make the data comparable across countries, we take the starting date $t=1$ to be the day of the $100$th confirmed case in each country. This way of ``normalizing'' the data is common practice \citep[cp.\ e.g.][]{Cohen2020}. A simple way to model the count data $\X_{it}$ is to use a Poisson distribution. Specifically, we may assume that the random variables $\X_{it}$ are Poisson distributed with time-varying intensity parameter $\lambda_i(t/T)$, that is, $\X_{it} \sim P_{\lambda_i(t/T)}$. Since $\lambda_i(t/T) = \ex[\X_{it}] = \var(\X_{it})$, we can model the observations $\X_{it}$ by the nonparametric regression equation 
\begin{equation}\label{eq:model-intro}
\X_{it} = \lambda_i\Big(\frac{t}{T}\Big) + u_{it} 
\end{equation}
for $1 \le t \le T$, where $u_{it} = \X_{it} - \ex[\X_{it}]$ with $\ex[u_{it}] = 0$ and $\var(u_{it}) = \lambda_i(t/T)$. As usual in nonparametric regression \citep[cp.][]{Robinson1989}, we let the regression function $\lambda_i$ in model \eqref{eq:model-intro} depend on rescaled time $t/T$ rather than on real time $t$. Hence, $\lambda_i: [0,1] \rightarrow \reals$ can be regarded as a function on the unit interval, which allows us to estimate it by standard techniques from nonparametric regression. Since $\lambda_i$ is a function of rescaled time $t/T$, the variables $\X_{it}$ in model \eqref{eq:model-intro} depend on the time series length $T$ in general, that is, $\X_{it} = \X_{it,T}$. To keep the notation simple, we however suppress this dependence throughout the paper. In Section \ref{sec:model}, we introduce the model setting in detail which underlies our analysis. As we will see there, it is a generalized version of the Poisson model \eqref{eq:model-intro}.

In model \eqref{eq:model-intro}, the time trend of new COVID-19 infections in country $i$ is described by the intensity function $\lambda_i$ of the underlying Poisson distribution. Hence, the question whether the time trends are comparable across countries amounts to the question whether the intensity functions $\lambda_i$ have the same shape across countries $i$. In this paper, we construct a multiscale test which allows to \textit{identify} and \textit{locate} the differences between the functions $\lambda_i$. More specifically, let $\intervals = \{ \mathcal{I}_k \subseteq[0,1] : 1 \le k \le K \}$ be a family of (rescaled) time intervals $\mathcal{I}_k$ and let $H_0^{(ijk)}$ be the hypothesis that the functions $\lambda_i$ and $\lambda_j$ are the same on the interval $\mathcal{I}_k$, that is, 
\[ H_0^{(ijk)}: \lambda_i(w) = \lambda_j(w) \text{ for all } w \in \mathcal{I}_k. \]
We design a method to test the hypothesis $H_0^{(ijk)}$ \textit{simultaneously} for all pairs of countries $i$ and $j$ under consideration and for all intervals $\mathcal{I}_k$ in the family $\intervals$. The main theoretical result of the paper shows that the method controls the familywise error rate, that is, the probability of wrongly rejecting at least one null hypothesis $H_0^{(ijk)}$. As we will see, this allows us to make simultaneous confidence statements of the following form for a given significance level $\alpha \in (0,1)$: 
\begin{center}
\begin{minipage}[c][1.25cm][c]{13cm}
\textit{With probability at least $1-\alpha$, the functions $\lambda_i$ and $\lambda_j$ differ on the interval $\mathcal{I}_k$ for every $(i,j,k)$ for which the test rejects $H_0^{(ijk)}$.} 
\end{minipage}
\end{center}
Hence, the method allows us to make simultaneous confidence statements (a) about which time trend functions differ from each other and (b) about where, that is, in which time intervals $\mathcal{I}_{k}$ they differ.

Even though our multiscale test is motivated by the current COVID-19 crisis, its applicability is by no means restricted to this specific event. It is a general method to compare nonparametric trends in epidemiological (count) data. It thus contributes to the literature on statistical tests for equality of nonparametric regression and trend curves. Examples of such tests can be found in \cite{HaerdleMarron1990}, \cite{Hall1990}, \cite{King1991}, \cite{Delgado1993}, \cite{Kulasekera1995}, \cite{YoungBowman1995}, \cite{MunkDette1998}, \cite{Lavergne2001}, \cite{NeumeyerDette2003} and \cite{Pardo-Fernandez2007}. More recent approaches were developed in \cite{DegrasWu2012}, \cite{Zhang2012}, \cite{Hidalgo2014} and \cite{ChenWu2018}. Compared to existing methods, our test has the following crucial advantage: it is much more informative. Most existing procedures allow to test \textit{whether} the regression or trend curves under consideration are all the same or not. However, they do not allow to infer \textit{which} curves are different and \textit{where} (that is, in which parts of the support) they differ. Our multiscale approach, in contrast, conveys this information. Indeed, it even allows to make rigorous confidence statements about which curves $\lambda_i$ are different and where they differ. To the best of our knowledge, there is no other method available in the literature which allows to make such simultaneous confidence statements. As far as we know, the only other multiscale test for comparing trend curves has been developed in \cite{Park2009}. However, their analysis is mainly methodological and not backed up by a general theory. In particular, theory is only available for the special case $n = 2$. Moreover, the theoretical results are only valid under very severe restrictions on the family of time intervals $\mathcal{F}$.

The paper is structured as follows. As already mentioned above, Section \ref{sec:model} details the model setting which underlies our analysis. The multiscale test is developed step by step in Section \ref{sec:test}. To keep the presentation as clear as possible, the technical details are deferred to the Appendix and the Supplementary Material. Section \ref{sec:empirics} contains the empirical part of the paper. There, we run some simulation experiments to demonstrate that the multiscale test has the formal properties predicted by the theory. Moreover, we use the test to compare the outbreak patterns of the COVID-19 epidemic in a number of European countries. 

\section{Model setting}\label{sec:model}

As already discussed in the Introduction, the assumption that $\X_{it} \sim P_{\lambda_i(t/T)}$ leads to a nonparametric regression model of the form 
\begin{equation}\label{eq:model-Poisson}
\X_{it} = \lambda_i\Big(\frac{t}{T}\Big) + u_{it} \quad \text{with} \quad u_{it} = \sqrt{\lambda_i\Big(\frac{t}{T}\Big)} \eta_{it}, 
\end{equation}
where $\eta_{it}$ has zero mean and unit variance. In this model, both the mean and the variance are described by the same function $\lambda_i$. In empirical applications, however, the variance often tends to be much larger than the mean. To deal with this issue, which has been known for a long time in the literature \citep{Cox1983} and which is commonly called overdispersion, so-called quasi-Poisson models \citep{McCullagh1989, Efron1986} are frequently used. In our context, a quasi-Poisson model of $\X_{it}$ has the form 
\begin{equation}\label{eq:model}
\X_{it} = \lambda_i\Big(\frac{t}{T}\Big) + \varepsilon_{it} \quad \text{with} \quad \varepsilon_{it} = \sigma \sqrt{\lambda_i\Big(\frac{t}{T}\Big)} \eta_{it}, 
\end{equation}
where $\sigma$ is a scaling factor that allows the variance to be a multiple of the mean function $\lambda_i$. In what follows, we assume that the observed data $\X_{it}$ are produced by model \eqref{eq:model}, where the noise residuals $\eta_{it}$ have zero mean and unit variance but we do not impose any further distributional assumptions on them.

Poisson and quasi-Poisson models are often used in the literature on epidemic modelling. \cite{De2020}, for example, assume that the observed COVID-19 case count in country $i$ follows a Poisson distribution with parameter $\lambda_i$ being a linear function of some covariate $Z_i$, that is, $\lambda_i = \beta Z_i$. \cite{Pellis2020} consider a quasi-Poisson model for the number of new COVID-19 cases. They in particular examine (a) a version of the model where the mean function is parametrically restricted to be exponentially growing with a constant growth rate and (b) a version where the mean function is modelled nonparametrically by splines. \cite{Tobias2020} analyze data on the accumulated number of cases using quasi-Poisson regression, where the mean function is modelled parametrically as a piecewise linear curve with known change points.

In order to derive our theoretical results, we impose the following regularity conditions on model \eqref{eq:model}:
\begin{enumerate}[label=(C\arabic*),leftmargin=1.0cm]
\item \label{C2} The functions $\lambda_i$ are uniformly Lipschitz continuous, that is, $|\lambda_i(u) - \lambda_i(v)| \le L |u-v|$ for all $u, v \in [0,1]$, where the constant $L$ does not depend on $i$. Moreover, they are uniformly bounded away from zero and infinity, that is, there exist constants $\lambda_{\min}$ and $\lambda_{\max}$ with $0 \le \lambda_{\min} \le \min_{w \in [0,1]} \lambda_i(w) \le \max_{w \in [0,1]} \lambda_i(w) \le \lambda_{\max} < \infty$ for all $i$. 
\item \label{C1} The random variables $\eta_{it}$ are independent both across $i$ and $t$. Moreover, for any $i$ and $t$, it holds that $\ex[\eta_{it}] = 0$, $\ex[\eta_{it}^2] = 1$ and $\ex[|\eta_{it}|^\theta] \le C_\theta < \infty$ for some $\theta > 4$. 
\end{enumerate}
\pagebreak
\ref{C2} imposes some standard-type regularity conditions on the functions $\lambda_i$. In particular, the functions are assumed to be smooth, bounded from above and bounded away from zero. The latter restriction is required because the noise variance in model \eqref{eq:model} equals zero if $\lambda_i$ is equal to zero. Since we normalize our test statistics by an estimate of the noise variance as detailed in Section \ref{sec:test}, we need this variance and thus the functions $\lambda_i$ to be bounded away from zero. \ref{C1} assumes the noise terms $\eta_{it}$ to fulfill some mild moment conditions and to be independent both across countries $i$ and time $t$. In the current COVID-19 crisis, independence across countries $i$ seems to be a fairly reasonable assumption due to severe travel restrictions, the closure of borders, etc. Independence across time $t$ is more debatable, but it is by no means unreasonable in our model framework: The time series process $\mathcal{\X}_i = \{X_{it}: 1 \le t \le T\}$ produced by model \eqref{eq:model} is nonstationary for each $i$. Specifically, both the mean $\ex[X_{it}] = \lambda_i(t/T)$ and the variance $\var(X_{it}) = \sigma^2 \lambda_i(t/T)$ are time-varying. A well-known fact in the time series literature is that nonstationarities such as a time-varying mean may produce spurious sample autocorrelations \citep[cp.\ e.g.][]{MikoschStarica2004, Fryzlewicz2008}. Hence, the observed persistence of a time series (captured by the sample autocorrelation function) may be due to nonstationarities rather than real autocorrelations. This insight has led researchers to prefer simple nonstationary models over intricate stationary time series models in some application areas such as finance \citep[cp.][]{MikoschStarica2000, MikoschStarica2004, Fryzlewicz2006, HafnerLinton2010}. In a similar vein, our model accounts for the persistence in the observed time series $\mathcal{\X}_i$ via nonstationarities rather than autocorrelations in the error terms.

\section{The multiscale test}\label{sec:test}

Let $\pairs \subseteq \{ (i,j): 1 \le i < j \le n \}$ be the set of all pairs of countries $(i,j)$ whose trend functions $\lambda_i$ and $\lambda_j$ we want to compare. Moreover, as already introduced above, let $\intervals = \{ \mathcal{I}_k: 1 \le k \le K \}$ be the family of (rescaled) time intervals under consideration. Finally, write $\indexset := \pairs \times \{1,\ldots,K\}$ and let $p := |\indexset|$ be the cardinality of $\indexset$. In this section, we devise a method to test the null hypothesis $H_0^{(ijk)}$ simultaneously for all pairs of countries $(i,j) \in \pairs$ and all time intervals $\mathcal{I}_k \in \intervals$, that is, for all $(i,j,k) \in \indexset$. The value $p = |\indexset|$ is the dimensionality of the simultaneous test problem we are dealing with. It amounts to the number of tests that we carry out simultaneously. As shown by our theoretical results in the Appendix, $p$ may be much larger than the time series length $T$, which means that the simultaneous test problem under consideration can be very high-dimensional.

\subsection{Construction of the test statistics}\label{subsec:test:test}

A statistic to test the hypothesis $H_0^{(ijk)}$ for a given triple $(i,j,k)$ can be constructed as follows. To start with, we introduce the expression 
\[ \hat{s}_{ijk,T} = \frac{1}{\sqrt{Th_k}} \sum\limits_{t=1}^T \ind\Big(\frac{t}{T} \in \mathcal{I}_k\Big) (\X_{it} - \X_{jt}), \]
where $h_k$ is the length of the time interval $\mathcal{I}_k$, $\ind(\cdot)$ denotes the indicator function and $\ind(t/T \in \mathcal{I}_k)$ can be regarded as a rectangular kernel weight. A simple application of the law of large numbers yields that $\hat{s}_{ijk,T}/\sqrt{Th_k} = (Th_k)^{-1} \sum_{t=1}^T \ind(t/T \in \mathcal{I}_k) \{\lambda_i(t/T) - \lambda_j(t/T)\} + o_p(1)$ for any fixed pair of countries $(i,j)$. Hence, the statistic $\hat{s}_{ijk,T}/\sqrt{Th_k}$ estimates the average distance between the functions $\lambda_i$ and $\lambda_j$ on the interval $\mathcal{I}_k$. Under \ref{C1}, it holds that
\begin{align*}
\nu_{ijk,T}^2 := \var(\hat{s}_{ijk,T}) 
 & = \frac{\sigma^2}{Th_k} \sum\limits_{t=1}^T \ind\Big(\frac{t}{T} \in \mathcal{I}_k\Big) \Big\{ \lambda_i\Big(\frac{t}{T}\Big) + \lambda_j\Big(\frac{t}{T}\Big) \Big\}. 
\end{align*}
In order to normalize the variance of the statistic $\hat{s}_{ijk,T}$, we scale it by an estimator of $\nu_{ijk,T}$. In particular, we estimate $\nu_{ijk,T}^2$ by 
\[ \hat{\nu}_{ijk,T}^2 = \frac{\hat{\sigma}^2}{Th_k} \sum\limits_{t=1}^T \ind\Big(\frac{t}{T} \in \mathcal{I}_k\Big) \{ \X_{it} + \X_{jt} \}, \]
where $\hat{\sigma}^2$ is defined as follows: For each country $i$, let 
\begin{align*}
\hat{\sigma}_i^2 = \frac{\sum_{t=2}^T (\X_{it}-\X_{it-1})^2}{2 \sum_{t=1}^T \X_{it}}
\end{align*}
and set $\hat{\sigma}^2 = |\countries|^{-1} \sum_{i \in \countries} \hat{\sigma}_i^2$ with $\countries = \{ \ell: \ell = i$ or $\ell = j$ for some $(i,j) \in \pairs \}$ denoting the set of countries that are taken into account by our test. The idea behind the estimator $\hat{\sigma}_i^2$ is as follows: Since $\lambda_i$ is Lipschitz continuous,  
\[ \X_{it} - \X_{it-1} = \sigma \sqrt{\lambda_i\Big(\frac{t}{T}\Big)} (\eta_{it} - \eta_{it-1}) + r_{it}, \]
where $|r_{it}| \le C(1+|\eta_{it-1}|)/T$ with a sufficiently large constant $C$. This suggests that $T^{-1} \sum_{t=2}^T (\X_{it} - \X_{it-1})^2 = 2 \sigma^2 \{ T^{-1} \sum_{t=2}^T \lambda_i(t/T) \} + o_p(1)$. Moreover, since $T^{-1} \sum_{t=1}^T \X_{it} = T^{-1} \sum_{t=1}^T \lambda_i(t/T) + o_p(1)$, we expect that $\hat{\sigma}_i^2 = \sigma^2 + o_p(1)$ for any $i$ and thus $\hat{\sigma}^2 = \sigma^2 + o_p(1)$. In Lemma \ref{lemmaS1} of the Supplementary Material, we formally show that $\hat{\sigma}^2$ is a consistent estimator of $\sigma^2$ under our regularity conditions. Normalizing the statistic $\hat{s}_{ijk,T}$ by the estimator $\hat{\nu}_{ijk,T}$ yields the expression 
\begin{equation}\label{eq:stat}
\hat{\psi}_{ijk,T} := \frac{\hat{s}_{ijk,T}}{\hat{\nu}_{ijk,T}} = \frac{\sum\nolimits_{t=1}^T \ind(\frac{t}{T} \in \mathcal{I}_k) (\X_{it} - \X_{jt})}{ \hat{\sigma} \{ \sum\nolimits_{t=1}^T \ind(\frac{t}{T} \in \mathcal{I}_k) (\X_{it} + \X_{jt}) \}^{1/2}}, 
\end{equation}
which serves as our test statistic of the hypothesis $H_0^{(ijk)}$. For later reference, we additionally introduce the statistic 
\begin{equation}\label{eq:stat0}
\hat{\psi}_{ijk,T}^0 = \frac{\sum\nolimits_{t=1}^T \ind(\frac{t}{T} \in \mathcal{I}_k) \, \sigma \overline{\lambda}_{ij}^{1/2}(\frac{t}{T}) (\eta_{it} - \eta_{jt})}{ \hat{\sigma} \{ \sum\nolimits_{t=1}^T \ind(\frac{t}{T} \in \mathcal{I}_k) (\X_{it} + \X_{jt}) \}^{1/2}}
\end{equation}
with $\overline{\lambda}_{ij}(u) = \{ \lambda_i(u) + \lambda_j(u) \}/2$, which is identical to $\hat{\psi}_{ijk,T}$ under $H_0^{(ijk)}$.

\subsection{Construction of the test}

Our multiscale test is carried out as follows: For a given significance level $\alpha \in (0,1)$ and each $(i,j,k) \in \indexset$, we reject $H_0^{(ijk)}$ if 
\[ |\hat{\psi}_{ijk,T}| > c_{ijk,T}(\alpha), \]
where $c_{ijk,T}(\alpha)$ is the critical value for the $(i,j,k)$-th test problem. The critical values $c_{ijk,T}(\alpha)$ are chosen such that the familywise error rate (FWER) is controlled at level $\alpha$, which is defined as the probability of wrongly rejecting $H_0^{(ijk)}$ for at least one $(i,j,k)$. More formally speaking, for a given significance level $\alpha \in (0,1)$, the FWER is 
\begin{align*}
\text{FWER}(\alpha) 
 & = \pr \Big( \exists (i,j,k) \in \indexset_0: |\hat{\psi}_{ijk,T}| > c_{ijk,T}(\alpha) \Big) \\
 & =  1 - \pr \Big( \forall (i,j,k) \in \indexset_0: |\hat{\psi}_{ijk,T}| \le c_{ijk,T}(\alpha) \Big) \\
 & = 1 - \pr\Big( \max_{(i,j,k) \in \indexset_0} |\hat{\psi}_{ijk,T}| \le c_{ijk,T}(\alpha) \Big), 
\end{align*}
where $\indexset_0 \subseteq \indexset$ is the set of triples $(i,j,k)$ for which $H_0^{(ijk)}$ holds true.

There are different ways to construct critical values $c_{ijk,T}(\alpha)$ that ensure control of the FWER at level $\alpha$. In the traditional approach, the same critical value $c_T(\alpha) = c_{ijk,T}(\alpha)$ is used for all $(i,j,k)$. In this case, controlling the FWER at the level $\alpha$ requires to determine the critical value $c_T(\alpha)$ such that
\begin{equation}\label{eq:FWER-tilde}
\text{FWER}(\alpha) = 1 - \pr\Big( \max_{(i,j,k) \in \indexset_0} |\hat{\psi}_{ijk,T}| \le c_T(\alpha) \Big) \le \alpha. 
\end{equation}
This can be achieved by choosing $c_T(\alpha)$ as the $(1-\alpha)$-quantile of the statistic 
\[ \tilde{\Psi}_T = \max_{(i,j,k) \in \indexset} |\hat{\psi}_{ijk,T}^0|, \]
where $\hat{\psi}_{ijk,T}^0$ was introduced in \eqref{eq:stat0}. (Note that both the statistic $\tilde{\Psi}_T$ and the quantile $c_T(\alpha)$ depend on the dimensionality $p$ of the test problem in general. To keep the notation simple, we however suppress this dependence throughout the paper. We use the same convention for all other quantities that are defined in the sequel.)

A more modern approach assigns different critical values $c_{ijk,T}(\alpha)$ to the test problems $(i,j,k)$. In particular, the critical value for the hypothesis $H_0^{(ijk)}$ is allowed to depend on the length $h_k$ of the time interval $\mathcal{I}_k$, that is, on the scale of the test problem. A general approach to construct scale-dependent critical values was pioneered by \cite{DuembgenSpokoiny2001} and has been used in many other studies since then; cp.\ for example \cite{Rohde2008}, \cite{DuembgenWalther2008}, \cite{RufibachWalther2010}, \cite{SchmidtHieber2013}, \cite{EckleBissantzDette2017} and \cite{Dunker2019}. In our context, the approach of \cite{DuembgenSpokoiny2001} leads to the critical values 
\begin{equation*}
c_{ijk,T}(\alpha) = c_T(\alpha,h_k) := b_k + q_T(\alpha)/a_k, 
\end{equation*}
where $a_k = \{\log(e/h_k)\}^{1/2} / \log \log(e^e / h_k)$ and $b_k = \sqrt{2 \log(1/h_k)}$ are scale-dependent constants and the quantity $q_T(\alpha)$ is determined by the following consideration: Since 
\begin{align}
\text{FWER}(\alpha)
  & = \pr \Big( \exists (i,j,k) \in \indexset_0: |\hat{\psi}_{ijk,T}| > c_T(\alpha,h_k) \Big) \nonumber \\
 & =  1 - \pr \Big( \forall (i,j,k) \in \indexset_0: |\hat{\psi}_{ijk,T}| \le c_T(\alpha,h_k) \Big) \nonumber \\
 & =  1 - \pr \Big( \forall (i,j,k) \in \indexset_0: a_k \big(|\hat{\psi}_{ijk,T}| - b_k\big) \le q_T(\alpha) \Big) \nonumber \\
 & = 1 - \pr\Big( \max_{(i,j,k) \in \indexset_0} a_k \big( |\hat{\psi}_{ijk,T}| - b_k \big) \le q_T(\alpha) \Big), \label{eq:FWER-hat}
\end{align}
we need to choose the quantity $q_T(\alpha)$ as the $(1-\alpha)$-quantile of the statistic 
\[ \hat{\Psi}_T = \max_{(i,j,k) \in \indexset} a_k \big( |\hat{\psi}_{ijk,T}^0| - b_k \big) \]
in order to ensure control of the FWER at level $\alpha$. Comparing \eqref{eq:FWER-hat} with \eqref{eq:FWER-tilde}, the current approach can be seen to differ from the traditional one in the following respect: the maximum statistic $\tilde{\Psi}_T$ is replaced by the rescaled version $\hat{\Psi}_T$ which re-weights the individual statistics $\hat{\psi}_{ijk,T}^0$ by the scale-dependent constants $a_k$ and $b_k$. As demonstrated above, this translates into scale-dependent critical values $c_{ijk,T}(\alpha) = c_T(\alpha,h_k)$.

Our theory allows us to work with both the traditional choice $c_{ijk,T}(\alpha) = c_T(\alpha)$ and the more modern, scale-dependent choice $c_{ijk,T}(\alpha) = c_T(\alpha,h_k)$. Since the latter choice produces a test approach with better theoretical properties in general \citep[cp.][]{DuembgenSpokoiny2001}, we restrict attention to the critical values $c_T(\alpha,h_k)$ in the sequel. There is, however, one complication we need to deal with: As the quantiles $q_T(\alpha)$ are not known in practice, we cannot compute the critical values $c_T(\alpha,h_k)$ exactly in practice but need to approximate them. This can be achieved as follows: Under appropriate regularity conditions, it can be shown that 
\begin{align*}
\hat{\psi}_{ijk,T}^0 
 & = \frac{\sum\nolimits_{t=1}^T \ind(\frac{t}{T} \in \mathcal{I}_k) \, \sigma \overline{\lambda}_{ij}^{1/2}(\frac{t}{T}) (\eta_{it} - \eta_{jt})}{ \hat{\sigma} \{ \sum\nolimits_{t=1}^T \ind(\frac{t}{T} \in \mathcal{I}_k) (\X_{it} + \X_{jt}) \}^{1/2}} \\
 & \approx \frac{1}{\sqrt{2Th_k}} \sum\limits_{t=1}^T \ind\Big(\frac{t}{T} \in \mathcal{I}_k\Big) \{ \eta_{it} - \eta_{jt} \}.
\end{align*} 
A Gaussian version of the statistic displayed in the final line above is given by 
\begin{equation*}
\phi_{ijk,T} = \frac{1}{\sqrt{2Th_k}} \sum\limits_{t=1}^T \ind\Big(\frac{t}{T} \in \mathcal{I}_k\Big) \big\{ Z_{it} - Z_{jt} \big\},
\end{equation*}
where $Z_{it}$ are independent standard normal random variables for $1 \le t \le T$ and $1 \le i \le n$. Hence, the statistic 
\[ \Phi_T = \max_{(i,j,k) \in \indexset} a_k \big( |\phi_{ijk,T}| - b_k \big) \]
can be regarded as a Gaussian version of the statistic $\hat{\Psi}_T$. We approximate the unknown quantile $q_T(\alpha)$ by the $(1-\alpha)$-quantile $q_{T,\text{Gauss}}(\alpha)$ of $\Phi_T$, which can be computed (approximately) by Monte Carlo simulations and can thus be treated as known.

To summarize, we propose the following procedure to simultaneously test the hypothesis $H_0^{(ijk)}$ for all $(i,j,k) \in \indexset$ at the significance level $\alpha \in (0,1)$: 
\begin{equation}\label{eq:test}
\text{For each } (i,j,k) \in \indexset, \text{ reject } H_0^{(ijk)} \text{ if } |\hat{\psi}_{ijk,T}| > c_{T,\text{Gauss}}(\alpha,h_k),
\end{equation}
where $c_{T,\text{Gauss}}(\alpha,h_k) = b_k + q_{T,\text{Gauss}}(\alpha)/a_k$ with $a_k = \{\log(e/h_k)\}^{1/2} / \log \log(e^e / h_k)$ and $b_k = \sqrt{2 \log(1/h_k)}$.

\subsection{Formal properties of the test}\label{subsec:test:properties}

In Theorem \ref{theo1} of the Appendix, we prove that under appropriate regularity conditions, the test defined in \eqref{eq:test} (asymptotically) controls the familywise error rate $\text{FWER}(\alpha)$ for each pre-specified significance level $\alpha$. As shown in Corollary \ref{corollary1}, this has the following implication: 
\begin{align}
\pr\Big( \forall (i,j,k) \in \indexset: \text{ If } |\hat{\psi}_{ijk,T}| > c_{T,\textnormal{Gauss}}(\alpha,h_k), \text{ then } & (i,j,k) \notin \indexset_0 \Big) \nonumber \\ & \ge 1 - \alpha + o(1), \label{eq:simconfstat}
\end{align} 
where $\indexset_0$ is the set of triples $(i,j,k) \in \indexset$ for which $H_0^{(ijk)}$ holds true. Verbally, \eqref{eq:simconfstat} can be expressed as follows:  
\begin{equation}\label{eq:confidencestatement1}
\begin{minipage}[c][1cm][c]{13cm}
\textit{With (asymptotic) probability at least $1-\alpha$, the null hypothesis $H_0^{(ijk)}$ is violated for all $(i,j,k) \in \indexset$ for which the test rejects $H_0^{(ijk)}$.} 
\end{minipage}
\end{equation}
In other words: 
\begin{equation}\label{eq:confidencestatement2}
\begin{minipage}[c][1cm][c]{13cm}
\textit{With (asymptotic) probability at least $1-\alpha$, the functions $\lambda_i$ and $\lambda_j$ differ on the interval $\mathcal{I}_k$ for all $(i,j,k) \in \indexset$ for which the test rejects $H_0^{(ijk)}$.} 
\end{minipage}
\end{equation}
Hence, the test allows us to make simultaneous confidence statements (a) about which pairs of countries $(i,j)$ have different trend functions and (b) about where, that is, in which time intervals $\mathcal{I}_k$ the functions differ.

\subsection{Implementation of the test in practice}

For a given significance level $\alpha \in (0,1)$, the test procedure defined in \eqref{eq:test} is implemented as follows in practice: 
\begin{enumerate}[label=\textit{Step \arabic*.}, leftmargin=1.45cm]
\item Compute the quantile $q_{T,\text{Gauss}}(\alpha)$ by Monte Carlo simulations. Specifically, draw a large number $N$ (say $N=5000$) samples of independent standard normal random variables $\{Z_{it}^{(\ell)} : 1 \le t \le T, \, 1 \le i \le n \}$ for $1 \le \ell \le N$. Compute the value $\Phi_T^{(\ell)}$ of the Gaussian statistic $\Phi_T$ for each sample $\ell$ and calculate the empirical $(1-\alpha)$-quantile $\hat{q}_{T,\text{Gauss}}(\alpha)$ from the values $\{ \Phi_T^{(\ell)}: 1 \le \ell \le N \}$. Use $\hat{q}_{T,\text{Gauss}}(\alpha)$ as an approximation of the quantile $q_{T,\text{Gauss}}(\alpha)$. 
\item Compute the critical values $c_{T,\text{Gauss}}(\alpha,h_k)$ for $1 \le k \le K$ based on the approximation $\hat{q}_{T,\text{Gauss}}(\alpha)$.
\item Carry out the test for each $(i,j,k) \in \indexset$ and store the test results in the variable $r_{ijk,T} = \ind( |\hat{\psi}_{ijk,T}| > c_{T,\text{Gauss}}(\alpha,h_k))$ for each $(i,j,k) \in \indexset$, that is, let $r_{ijk,T} = 1$ if the hypothesis $H_0^{(ijk)}$ is rejected and $r_{ijk,T} = 0$ otherwise. 
\end{enumerate}

To graphically present the test results, we produce a plot for each pair of countries $(i,j) \in \pairs$ that shows the intervals $\mathcal{I}_k$ for which the test rejects the null $H_0^{(ijk)}$, that is, the intervals in the set $\intervals_{\text{reject}}(i,j) = \{ \mathcal{I}_k \in \intervals: r_{ijk,T} = 1 \}$. The plot is designed such that it graphically highlights the subset of intervals $\intervals_{\text{reject}}^{\text{min}}(i,j) = \{ \mathcal{I}_k \in \intervals_{\text{reject}}(i,j):$ there exists no $\mathcal{I}_{k^\prime} \in \intervals_{\text{reject}}(i,j)$ with $\mathcal{I}_{k^\prime} \subset \mathcal{I}_k \}$. The elements of $\intervals_{\text{reject}}^{\text{min}}(i,j)$ are called minimal intervals. By definition, there is no other interval $\mathcal{I}_{k^\prime}$ in $\intervals_{\text{reject}}(i,j)$ which is a proper subset of a minimal interval $\mathcal{I}_k$. Hence, the minimal intervals can be regarded as those intervals in $\intervals_{\text{reject}}(i,j)$ which are most informative about the precise location of the differences between the trends $\lambda_i$ and $\lambda_j$. In Section \ref{sec:empirics}, we use the graphical device just described to present the test results of our empirical application; cp.\ panels (d) in Figures \ref{fig:Germany:Italy}--\ref{fig:Germany:UK}.

According to \eqref{eq:simconfstat}, we can make the following simultaneous confidence statement about the intervals in $\intervals_{\text{reject}}(i,j)$ for $(i,j) \in \pairs$: 
\begin{equation}\label{eq:confidencestatement3}
\begin{minipage}[c][1.75cm][c]{13cm}
\textit{With (asymptotic) probability at least $1-\alpha$, it holds that for every pair of countries $(i,j) \in \pairs$, the functions $\lambda_i$ and $\lambda_j$ differ on each interval in $\intervals_{\text{reject}}(i,j)$.} 
\end{minipage}
\end{equation}
Hence, we can claim with statistical confidence at least $1-\alpha$ that the functions $\lambda_i$ and $\lambda_j$ differ on each time interval which is depicted in the plots of our graphical device. Since $\intervals_{\text{reject}}^{\text{min}}(i,j) \subseteq \intervals_{\text{reject}}(i,j)$ for any $(i,j) \in \pairs$, the confidence statement \eqref{eq:confidencestatement3} trivially remains to hold true when the sets $\intervals_{\text{reject}}(i,j)$ are replaced by $\intervals_{\text{reject}}^{\text{min}}(i,j)$.

\section{Empirical application to COVID-19 data}\label{sec:empirics}

We now use our test to analyze the outbreak patterns of the COVID-19 epidemic. We proceed in two steps. In Section \ref{subsec:sim}, we assess the finite sample performance of our test by Monte-Carlo experiments. Specifically, we run a series of experiments which show that the test controls the FWER at level $\alpha$ as predicted by the theory and that it has good power properties. In Section \ref{subsec:app}, we then apply the test to a sample of COVID-19 data from different European countries. Our multiscale test is implemented in the R package \verb|multiscale|, available on GitHub at \texttt{https://github.com/marina-khi/multiscale}.

\subsection{Simulation experiments}\label{subsec:sim}

We simulate count data $\mathcal{\X} = \{ X_{it}: 1 \le i \le n, 1 \le t \le T \}$ by drawing the observations $X_{it}$ independently from a negative binomial distribution with mean $\lambda_i(t/T)$ and variance $\sigma^2 \lambda_i(t/T)$. By definition, $X_{it}$ has a negative binomial distribution with para\-meters $q$ and $r$ if $\pr(X_{it} = m) = \Gamma(m+r) / (\Gamma(r) m!) q^r (1-q)^m$ for each $m \in \naturals \cup \{0\}$. Since $\ex[X_{it}] = r (1-q)/q$ and $\var(X_{it}) = r(1-q)/q^2$, we can use the parametrization $q = 1/\sigma^2$ and $r = \lambda_i(t/T) / (\sigma^2 - 1)$ to obtain that $\ex[X_{it}] = \lambda_i(t/T)$ and $\var(X_{it}) = \sigma^2 \lambda_i(t/T)$. With this parametrization, the simulated data follow a nonparametric regression model of the form 
\[ \X_{it} = \lambda_i\Big(\frac{t}{T}\Big) + \sigma \sqrt{\lambda_i\Big(\frac{t}{T}\Big)} \eta_{it}, \]
where the noise variables $\eta_{it}$ have zero mean and unit variance. The functions $\lambda_i$ are specified below. The overdispersion parameter is set to $\sigma = 15$, which is similar to the estimate $\hat{\sigma} = 14.44$ obtained in the empirical application of Section \ref{subsec:app}. Robustness checks with $\sigma=10$ and $\sigma=20$ are provided in the Supplementary Material.

We consider different values for $T$ and $n$, in particular, $T \in \{100,250, 500\}$ and $n \in \{5,10,50\}$. Note that in the application, we have $T=139$ and $n=5$. We let $\pairs = \{ (i,j): 1 \le i < j \le n\}$, that is, we compare all pairs of countries $(i,j)$ with $i < j$. Moreover, we choose $\intervals$ to be a family of time intervals $\mathcal{I}_k$ with length $h_k \in \{ 7/T,14/T,21/T,28/T \}$. Hence, the intervals in $\intervals$ have length either $7$, $14$, $21$ or $28$ days (i.e., $1$, $2$, $3$ or $4$ weeks). For each length $h_k$, we include all intervals that start at days $t = 1 + 7(j-1)$ and $t = 4 + 7(j-1)$ for $j=1,2,\ldots$ A graphical presentation of the family $\intervals$ for $T = 139$ (as in the application) is given in Figure \ref{fig:intervals}. All our simulation experiments are based on $R=5000$ simulation runs.

\begin{figure}[t!]
\centering
\begin{subfigure}[b]{0.475\textwidth}
\includegraphics[width=\textwidth]{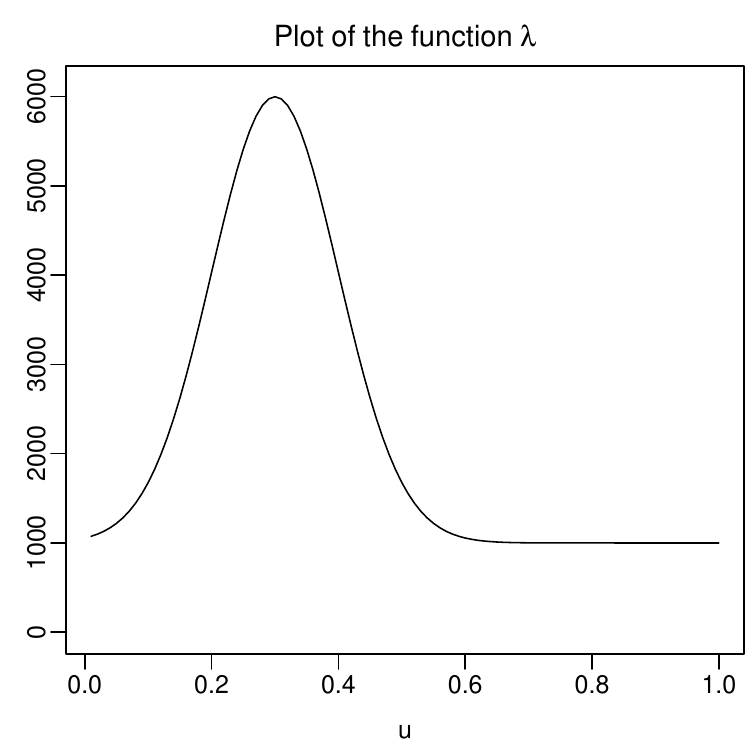}
\caption{}\label{fig:lambda}
\end{subfigure}\hspace{0.25cm}
\begin{subfigure}[b]{0.475\textwidth}
\includegraphics[width=\textwidth]{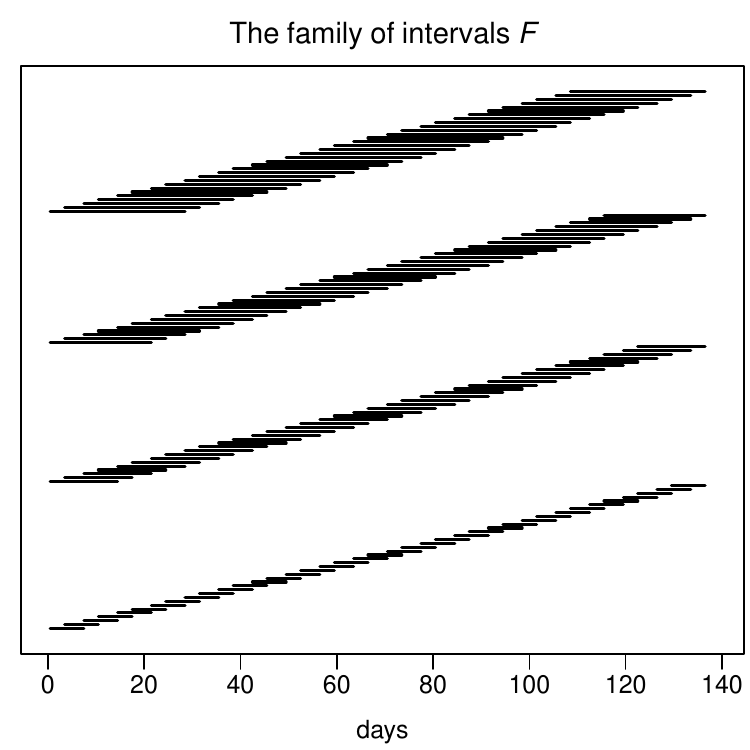}
\caption{}\label{fig:intervals}
\end{subfigure}
\caption{(a) Plot of the function $\lambda$; (b) plot of the family of intervals $\intervals$.}\label{fig:lambda_F}
\end{figure}

In the first part of the simulation study, we examine whether our test controls the FWER as predicted by the theory. To do so, we assume that the hypothesis $H_0^{(ijk)}$ holds true for all $(i,j,k)$ under consideration, which implies that $\lambda_i = \lambda$ for all $i$. We consider the function 
\begin{equation}\label{eq:lambda}
\lambda(u) = 5000 \exp\Big(-\frac{(10 u-3)^2}{2}\Big) + 1000, 
\end{equation}
which is similar in shape to some of the estimated trend curves in the application of Section \ref{subsec:app}. A plot of the function $\lambda$ is provided in Figure \ref{fig:lambda}. To evaluate whether the test controls the FWER at level $\alpha$, we compare the empirical size of the test with the target $\alpha$. The empirical size is computed as the precentage of simulation runs in which the test falsely rejects at least one null hypothesis $H_0^{(ijk)}$.

The simulation results are reported in Table \ref{tab:sim:size}. As can be seen, the empirical size gives a reasonable approximation to the target $\alpha$ in all scenarios under investigation, even though the size numbers have a slight downward bias. This bias gets larger as the number of time series $n$ increases, which reflects the fact that the test problem becomes more difficult for larger $n$. Already for $n=5$, the number $p$ of hypotheses to be tested is quite high, in particular, $p = 960, 2\,680, 5\,560$ for $T=100,250,500$. This number increases to $p=117\,600, 328\,300, 681\,100$ when $n=50$. Hence, the dimensionality and thus the complexity of the test problem increases considerably as $n$ gets larger. On first sight, it may seem astonishing that the downward bias does not diminish notably as the time series length $T$ increases. This, however, has a simple explanation: The interval lengths $h_k$ remain the same ($7$, $14$, $21$ or $28$ days) as $T$ increases, which implies that the effective sample size for computing the test statistics $\hat{\psi}_{ijk,T}$ does not change as well. To summarize, even though slightly conservative, the test controls the FWER quite accurately in the simulation setting at hand.

\begin{table}[t!]
\footnotesize{
\caption{Empirical size of the test for different values of $n$ and $T$.}\label{tab:sim:size}
\newcolumntype{C}[1]{>{\hsize=#1\hsize\centering\arraybackslash}X}
\newcolumntype{Z}{>{\centering\arraybackslash}X}
\begin{tabularx}{\textwidth}{l Z@{\hskip 6pt}Z@{\hskip 6pt}Z Z@{\hskip 6pt}Z@{\hskip 6pt}Z Z@{\hskip 6pt}Z@{\hskip 6pt}Z} 
\toprule
 & \multicolumn{3}{c}{$n = 5$} & \multicolumn{3}{c}{$n = 10$} & \multicolumn{3}{c}{$n = 50$} \\
\cmidrule[0.4pt]{2-4} \cmidrule[0.4pt]{5-7} \cmidrule[0.4pt]{8-10}
 & \multicolumn{3}{c}{significance level $\alpha$} &\multicolumn{3}{c}{significance level $\alpha$} & \multicolumn{3}{c}{significance level $\alpha$} \\
 & 0.01 & 0.05 & 0.1  &  0.01 & 0.05 & 0.1  &  0.01 & 0.05 & 0.1 \\
\cmidrule[0.4pt]{1-10}
$T = 100$ & 0.011 & 0.047 & 0.093 & 0.010 & 0.044 & 0.087 & 0.008 & 0.037 & 0.075 \\ 
  $T = 250$ & 0.009 & 0.047 & 0.091 & 0.009 & 0.046 & 0.087 & 0.008 & 0.035 & 0.069 \\ 
  $T = 500$ & 0.010 & 0.044 & 0.083 & 0.008 & 0.048 & 0.093 & 0.007 & 0.035 & 0.077 \\ 
\bottomrule
\end{tabularx}}
\end{table}

In the second part of the simulation study, we investigate the power properties of the test. To do so, we assume that $\lambda_i = \lambda$ for all $i > 1$ and that $\lambda_1 \neq \lambda$, where $\lambda$ is defined in \eqref{eq:lambda}. Hence, only the first mean function $\lambda_1$ is different from the others. This implies that the hypothesis $H_0^{(ijk)}$ holds true for all $(i, j, k)$ with $i > 1$ and $j > 1$, while there is at least one hypothesis $H_0^{(ijk)}$ with either $i = 1$ or $j = 1$ that does not hold true. We consider two different simulation scenarios. In Scenario A, the function $\lambda_1$ has the form 
\[ \lambda_1(u) = 6000 \exp\Big(-\frac{(10 u-3)^2}{2}\Big) + 1000 \]
and is plotted together with $\lambda$ in Figure \ref{fig:lambda_fcts_height}. As can be seen, the two functions $\lambda_1$ and $\lambda$ peak at the same point in time, but the peak of $\lambda_1$ is higher than that of $\lambda$. In Scenario B, we let
\[ \lambda_1(u) = 5000 \exp\Big(-\frac{(9 u-3)^2}{2}\Big) + 1000. \]
Figure \ref{fig:lambda_fcts_shift} shows that the peaks of $\lambda_1$ and $\lambda$ have the same height but are reached at different points in time. To evaluate the power properties of the test in Scenarios A and B, we compute the percentage of simulation runs where the test (i) correctly detects differences between $\lambda_1$ and at least one of the other mean functions and (ii) does not spuriously detect differences between the other mean functions. Put differently, we calculate the percentage of simulation runs where (i) the set $\intervals_{\text{reject}}(1, j)$ is non-empty at least for one $j \in \{2, \ldots, n\}$ and (ii) all other sets $\intervals_{\text{reject}}(i, j)$ with $2 \leq i < j \leq n$ are empty. We call this percentage number the (empirical) power of the test. We thus use the term ``power'' a bit differently than usual.

\begin{figure}[t!]
\centering
\begin{subfigure}[b]{0.475\textwidth}
\includegraphics[width=\textwidth]{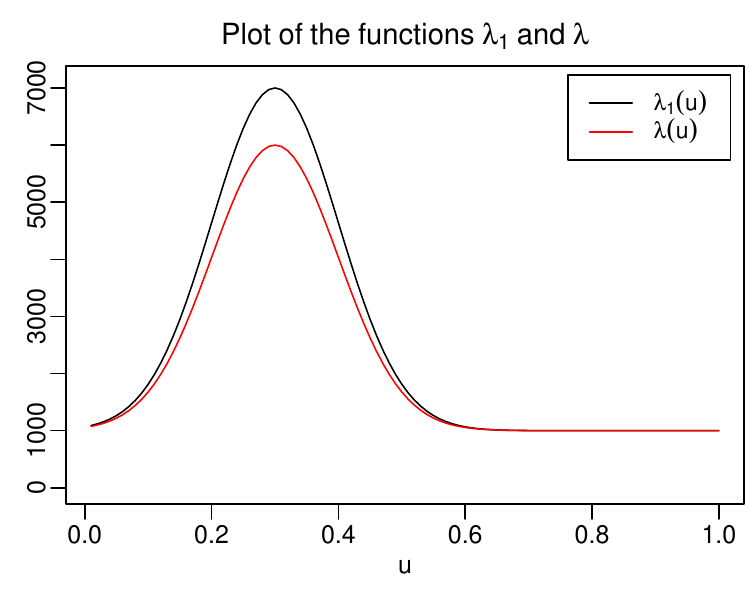}
\caption{Scenario A}\label{fig:lambda_fcts_height}
\end{subfigure}\hspace{0.25cm}
\begin{subfigure}[b]{0.475\textwidth}
\includegraphics[width=\textwidth]{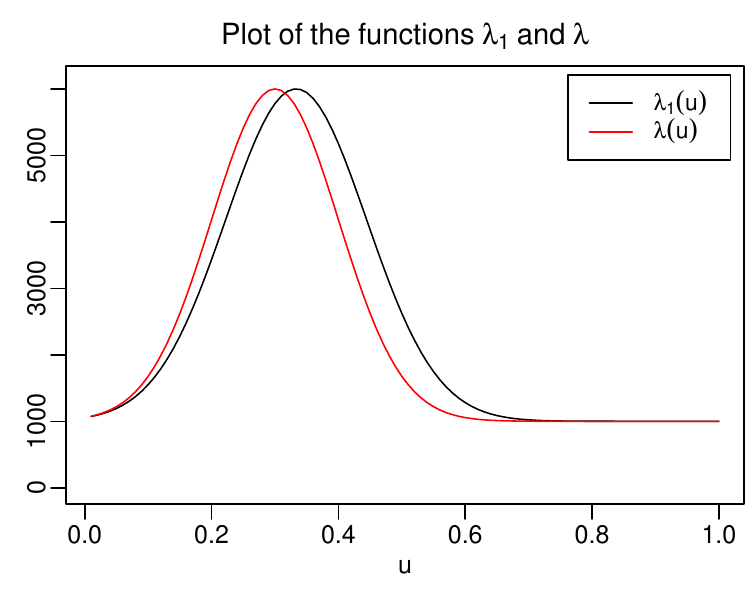}
\caption{Scenario B}\label{fig:lambda_fcts_shift}
\end{subfigure}
\caption{Plot of the functions $\lambda_1$ (black) and $\lambda$ (red) in the simulation scenarios A and B.}
\end{figure}

\begin{table}[t!]
\footnotesize{
\caption{Power of the test for different values of $n$ and $T$ in Scenario A.}\label{tab:sim:power:1}
\newcolumntype{C}[1]{>{\hsize=#1\hsize\centering\arraybackslash}X}
\newcolumntype{Z}{>{\centering\arraybackslash}X}
\begin{tabularx}{\textwidth}{l Z@{\hskip 6pt}Z@{\hskip 6pt}Z Z@{\hskip 6pt}Z@{\hskip 6pt}Z Z@{\hskip 6pt}Z@{\hskip 6pt}Z} 
\toprule
 & \multicolumn{3}{c}{$n = 5$} & \multicolumn{3}{c}{$n = 10$} & \multicolumn{3}{c}{$n = 50$} \\
\cmidrule[0.4pt]{2-4} \cmidrule[0.4pt]{5-7} \cmidrule[0.4pt]{8-10}
 & \multicolumn{3}{c}{significance level $\alpha$} &\multicolumn{3}{c}{significance level $\alpha$} & \multicolumn{3}{c}{significance level $\alpha$} \\
 & 0.01 & 0.05 & 0.1  &  0.01 & 0.05 & 0.1  &  0.01 & 0.05 & 0.1 \\
\cmidrule[0.4pt]{1-10}
$T = 100$ & 0.335 & 0.518 & 0.597 & 0.306 & 0.474 & 0.545 & 0.212 & 0.352 & 0.418 \\ 
  $T = 250$ & 0.615 & 0.790 & 0.836 & 0.580 & 0.764 & 0.800 & 0.470 & 0.648 & 0.705 \\ 
  $T = 500$ & 0.736 & 0.905 & 0.917 & 0.738 & 0.884 & 0.890 & 0.636 & 0.799 & 0.830 \\ 
\bottomrule
\end{tabularx}}
\end{table}

\begin{table}[t!]
\footnotesize{
\caption{Power of the test for different values of $n$ and $T$ in Scenario B.}\label{tab:sim:power:2}
\newcolumntype{C}[1]{>{\hsize=#1\hsize\centering\arraybackslash}X}
\newcolumntype{Z}{>{\centering\arraybackslash}X}
\begin{tabularx}{\textwidth}{l Z@{\hskip 6pt}Z@{\hskip 6pt}Z Z@{\hskip 6pt}Z@{\hskip 6pt}Z Z@{\hskip 6pt}Z@{\hskip 6pt}Z} 
\toprule
 & \multicolumn{3}{c}{$n = 5$} & \multicolumn{3}{c}{$n = 10$} & \multicolumn{3}{c}{$n = 50$} \\
\cmidrule[0.4pt]{2-4} \cmidrule[0.4pt]{5-7} \cmidrule[0.4pt]{8-10}
 & \multicolumn{3}{c}{significance level $\alpha$} &\multicolumn{3}{c}{significance level $\alpha$} & \multicolumn{3}{c}{significance level $\alpha$} \\
 & 0.01 & 0.05 & 0.1  &  0.01 & 0.05 & 0.1  &  0.01 & 0.05 & 0.1 \\
\cmidrule[0.4pt]{1-10}
$T = 100$ & 0.824 & 0.910 & 0.903 & 0.812 & 0.893 & 0.890 & 0.738 & 0.847 & 0.857 \\ 
  $T = 250$ & 0.991 & 0.972 & 0.941 & 0.991 & 0.960 & 0.920 & 0.991 & 0.965 & 0.933 \\ 
  $T = 500$ & 0.997 & 0.973 & 0.949 & 0.995 & 0.961 & 0.923 & 0.996 & 0.969 & 0.932 \\ 
\bottomrule
\end{tabularx}}
\end{table}

The results for Scenario A (see Figure \ref{fig:lambda_fcts_height}) are presented in Table \ref{tab:sim:power:1} and those for Scenario B (see Figure \ref{fig:lambda_fcts_shift}) in Table \ref{tab:sim:power:2}. As can be seen, the test has substantial power in all the considered simulation settings. It is more powerful in Scenario B than in Scenario A, which is most presumably due to the fact that the differences $|\lambda_1(u) - \lambda(u)|$ are much larger in Scenario B. Moreover, it is less powerful for larger numbers of time series $n$, which reflects the fact that the test problem gets more high-dimensional and thus more difficult as $n$ increases. As one would expect, the power numbers tend to become larger as the time series length $T$ and the significance level $\alpha$ increase. In Scenario B (mostly for $T=250$ and $T=500$), however, the power numbers drop down a bit as $\alpha$ gets larger. This reverse dependance can be explained by the way we calculate power: we exclude simulation runs where the test spuriously detects differences between the trends in countries $i$ and $j$ with $i, j > 1$. The number of spurious findings increases as we make the significance level $\alpha$ larger, which presumably causes the slight drop in power.

\subsection{Analysis of COVID-19 data}\label{subsec:app}

The COVID-19 pandemic is one of the most pressing issues at present. The first outbreak occurred in Wuhan, China, in December 2019. On 30 January 2020, the World Health Organization (WHO) declared that the outbreak constitutes a Public Health Emergency of International Concern, and on 11 March 2020, the WHO characterized it as a pandemic. As of 22 July 2020, more than $14.56$ million cases of COVID-19 infections have been reported worldwide, resulting in more than $607\,000$ deaths.

There are many open questions surrounding the current COVID-19 pandemic. A question which is particularly relevant for governments and policy makers is whether the pandemic has developed similarly in different countries or whether there are notable differences. Identifying these differences may give some insight into which government policies have been more effective in containing the virus than others. In what follows, we use our multiscale test to compare the development of COVID-19 in several European countries. It is important to emphasize that our test allows to identify differences in the development of the epidemic across countries in a statistically rigorous way, but it does not tell what causes these differences. By distinguishing statistically significant differences from artefacts of the sampling noise, the test provides the basis for a further investigation into the causes. Such an investigation, however, presumably goes beyond a mere statistical analysis.

\subsubsection{Data}

We analyze data from five European countries: Germany, Italy, Spain, France and the United Kingdom. For each country $i$, we observe a time series $\mathcal{X}_i = \{ X_{it}: 1 \le t \le T \}$, where $X_{it}$ is the number of newly confirmed COVID-19 cases in country $i$ on day $t$. The data are freely available on the homepage of the European Center for Disease Prevention and Control (\texttt{https://www.ecdc.europa.eu}) and were downloaded on 22 July 2020. As already mentioned in the Introduction, we take the day of the $100$th confirmed case in each country as the starting date $t=1$, which is a common way of ``normalizing'' the data and making them comparable across countries \citep[cp.][]{Cohen2020}. The time series length $T$ is taken to be the minimal number of days for which we have observations for all five countries. The resulting dataset consists of $n = 5$ time series, each with $T = 139$ observations (as of July 22). Some of the time series contain negative values which we replaced by $0$. Overall, this resulted in $6$ replacements. Plots of the observed time series are presented in the upper panels (a) of Figures \ref{fig:Germany:Italy}--\ref{fig:Germany:UK}.

To interpret the results produced by our multiscale test, we consider the Government Response Index (GRI) from the Oxford COVID-19 Government Response Tracker (OxCGRT) \citep{Hale2020}. The GRI measures how severe the actions are that are taken by a country's government to contain the virus. It is calculated based on several common government policies such as school closures and travel restrictions. The GRI ranges from $0$ to $100$, with $0$ corresponding to no response from the government at all and $100$ corresponding to full lockdown, closure of schools and workplaces, ban on travelling, etc. Detailed information on the collection of the data for government responses and the methodology for calculating the GRI is provided in \cite{Hale2020paper}. Plots of the GRI time series are given in panels (c) of Figures \ref{fig:Germany:Italy}--\ref{fig:Germany:UK}.

\subsubsection{Test results}

We assume that the data $X_{it}$ of each country $i$ in our sample follow the nonparametric trend model 
\[ \X_{it} = \lambda_i\Big(\frac{t}{T}\Big) + \sigma \sqrt{\lambda_i\Big(\frac{t}{T}\Big)} \eta_{it}, \]
which was introduced in equation \eqref{eq:model}. The overdispersion parameter $\sigma$ is estimated by the procedure described in Section \ref{subsec:test:test}, which yields the estimate $\hat{\sigma} = 14.44$. Throughout the section, we set the significance level to $\alpha=0.05$ and implement the multiscale test in exactly the same way as in the simulation study of Section \ref{subsec:sim}. In particular, we let $\pairs = \{ (i,j): 1 \le i < j \le 5\}$, that is, we compare all pairs of countries $(i,j)$ with $i < j$, and we choose $\intervals$ to be the family of time intervals plotted in Figure \ref{fig:intervals}. Hence, all intervals in $\intervals$ have length either 7, 14, 21 or 28 days.

With the help of our multiscale method, we simultaneously test the null hypo\-thesis $H_0^{(ijk)}$ that $\lambda_i = \lambda_j$ on the interval $\mathcal{I}_k$ for each $(i, j, k) \in \indexset$. The results are presented in Figures \ref{fig:Germany:Italy}--\ref{fig:Germany:UK}, each figure comparing a specific pair of countries $(i,j)$ from our sample. For the sake of brevity, we only show the results for the pairwise comparisons of Germany with each of the four other countries. The remaining figures can be found in Section \ref{s:subsec:app} of the Supplementary Material. Each figure splits into four panels (a)--(d).  Panel (a) shows the observed time series for the two countries $i$ and $j$ that are compared. Panel (b) presents smoothed versions of the time series from (a), that is, it shows nonparametric kernel estimates (specifically, Nadaraya-Watson estimates) of the two trend functions $\lambda_i$ and $\lambda_j$, where the bandwidth is set to $7$ days and a rectangular kernel is used. Panel (c) displays the Government Response Index (GRI) of the two countries. Finally, panel (d) presents the results produced by our test: it depicts in grey the set $\intervals_{\text{reject}}(i,j)$ of all the intervals $\mathcal{I}_k$ for which the test rejects the null $H_0^{(ijk)}$. The minimal intervals in the subset $\intervals_{\text{reject}}^{\text{min}}(i, j)$ are highlighted by a black frame. Note that according to \eqref{eq:simconfstat}, we can make the following simultaneous confidence statement about the intervals plotted in panels (d) of Figures \ref{fig:Germany:Italy}--\ref{fig:Germany:UK}: we can claim, with confidence of about $95\%$, that there is a difference between the functions $\lambda_i$ and $\lambda_j$ on each of these intervals.

\begin{figure}[p!]
\begin{minipage}[t]{0.49\textwidth}
\includegraphics[width=\textwidth]{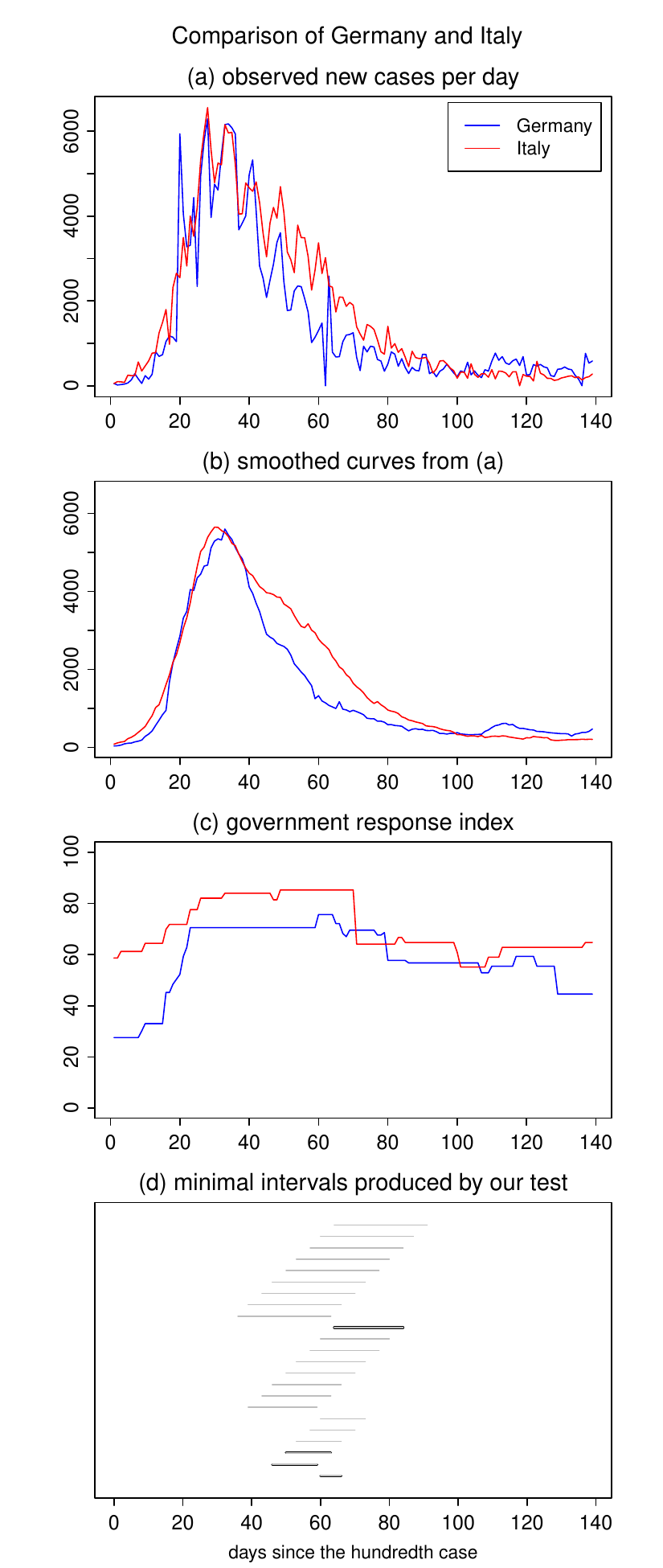}
\caption{Test results for the comparison of Germany and Italy.}\label{fig:Germany:Italy}
\end{minipage}
\hspace{0.25cm}
\begin{minipage}[t]{0.49\textwidth}
\includegraphics[width=\textwidth]{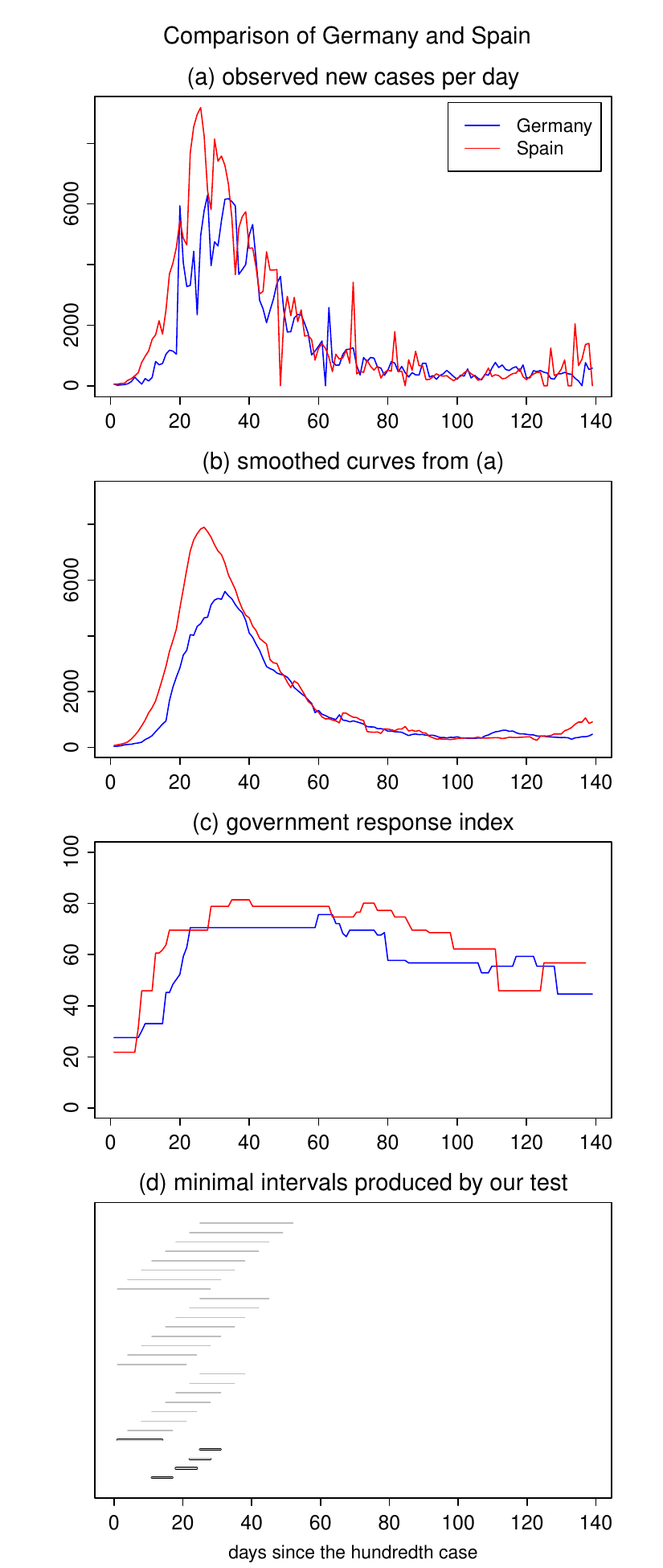}
\caption{Test results for the comparison of Germany and Spain.}\label{fig:Germany:Spain}
\end{minipage}

\caption*{Note: In each figure, panel (a) shows the two observed time series, panel (b) smoothed versions of the time series, and panel (c) the corresponding Government Response Index (GRI). Panel (d) depicts the set of intervals $\intervals_{\text{reject}}(i,j)$ in grey and the subset of minimal intervals $\intervals_{\text{reject}}^{\text{min}}(i,j)$ with a black frame.}
\end{figure}

\begin{figure}[p!]
\begin{minipage}[t]{0.49\textwidth}
\includegraphics[width=\textwidth]{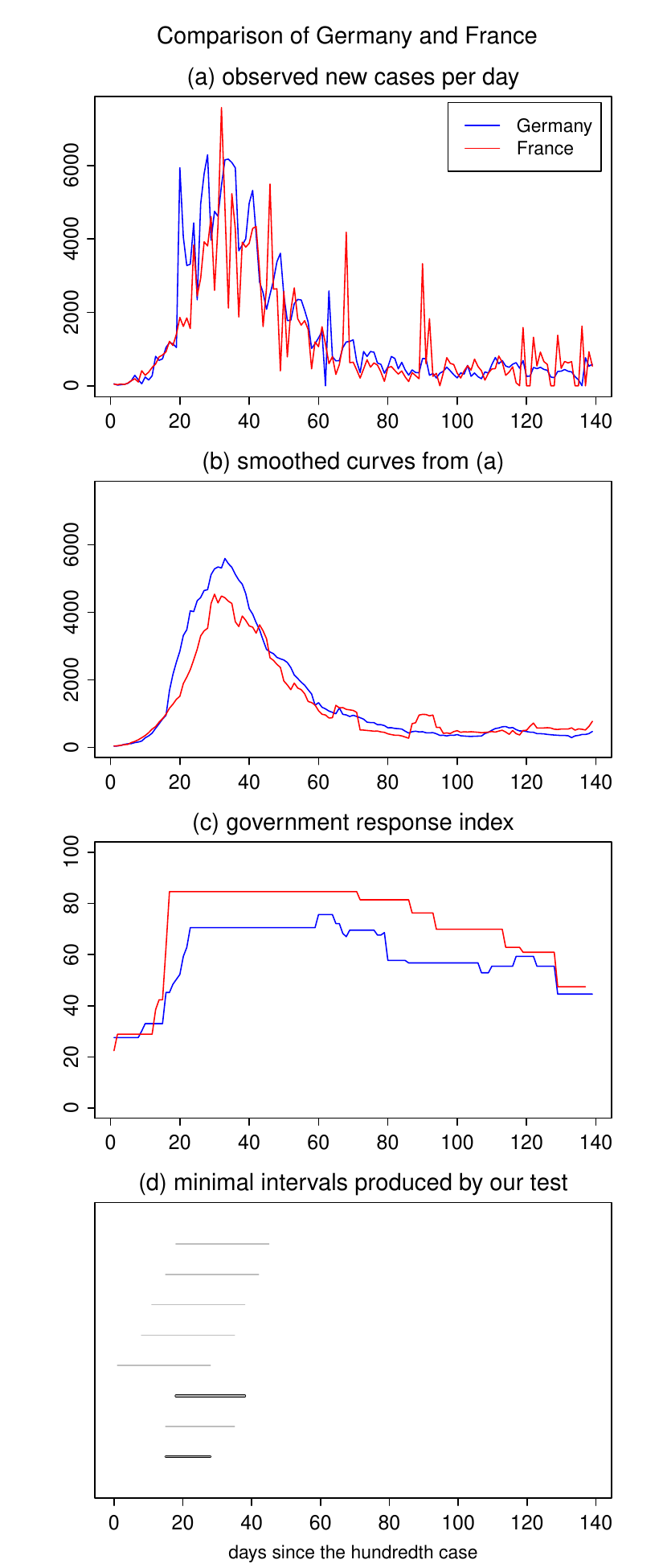}
\caption{Test results for the comparison of Germany and France.}\label{fig:Germany:France}
\end{minipage}
\hspace{0.25cm}
\begin{minipage}[t]{0.49\textwidth}
\includegraphics[width=\textwidth]{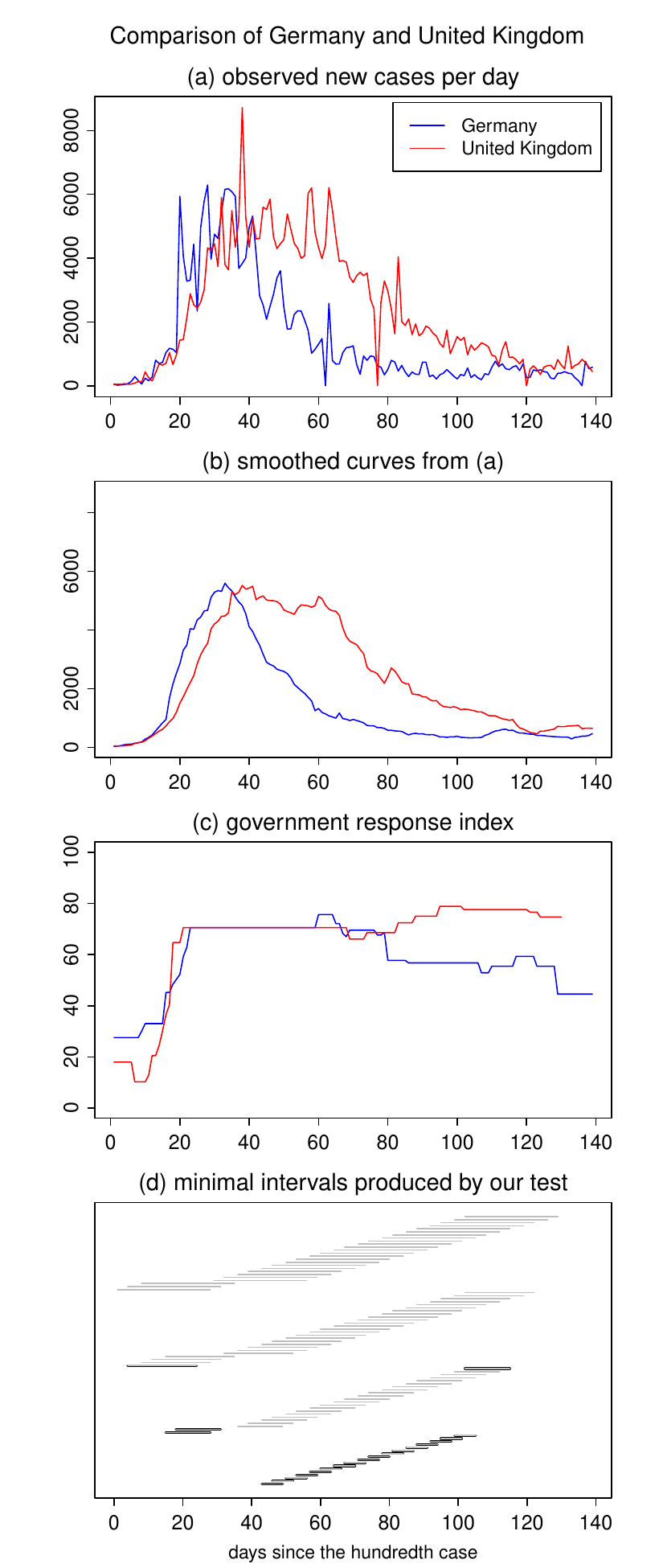}
\caption{Test results for the comparison of Germany and the UK.}\label{fig:Germany:UK}
\end{minipage}

\caption*{Note: In each figure, panel (a) shows the two observed time series, panel (b) smoothed versions of the time series, and panel (c) the corresponding Government Response Index (GRI). Panel (d) depicts the set of intervals $\intervals_{\text{reject}}(i,j)$ in grey and the subset of minimal intervals $\intervals_{\text{reject}}^{\text{min}}(i,j)$ with a black frame.}
\end{figure}

We now have a closer look at the results in Figures \ref{fig:Germany:Italy}--\ref{fig:Germany:UK}. Figure \ref{fig:Germany:Italy} presents the comparison of Germany with Italy. The two time series of daily new cases in panel (a) can be seen to be very similar until approximately day $40$. Thereafter, the German time series appears to trend downwards more strongly than the Italian one. The smoothed data in panel (b) give a similar visual impression: the kernel estimates of the German and Italian trend curves $\lambda_i$ and $\lambda_j$ are very close to each other until approximately day $40$ but then start to differ. It is however not clear whether the differences between the two curve estimates reflect differences between the underlying trend curves or whether these are mere artefacts of sampling noise. Our test allows to clarify this issue. Inspecting panel (d), we see that the test detects significant differences between the trend curves in the time period between day $36$ and $91$. However, it does not find any significant differences up to day $36$. 
Taken together, our results provide evidence that the epidemic developed very similarly in Germany and Italy until a peak was reached around day $40$. Thereafter, however, the German time series exhibits a significantly stronger downward trend than the Italian one.

Inspecting Figures \ref{fig:Germany:Spain} and  \ref{fig:Germany:France}, a quite different picture arises when comparing Germany with France and Spain. The test detects significant differences between the German trend and the trends in France and Spain up to (approximately) day $50$ but not thereafter. Hence, we find that the time trends evolve differently during the outbreak of the crisis, but they appear to decrease in more or less the same fashion after a peak was reached. Finally, the comparison of Germany with the UK in Figure \ref{fig:Germany:UK} reveals significant differences between the time trends over essentially the whole observation window. Inspecting the time series in panel (a), it is quite obvious that the UK trend evolves differently from the German one after day $40$. However, our test also detects differences between the trends during the onset of the crisis, which is not obvious from the time series plot in panel (a).

\subsubsection{Discussion}

Having identified significant differences between the epidemic trends in the five countries under consideration, one may ask next what are the causes of these differences. As already mentioned at the beginning of this section, this question cannot be answered by our test. Rather, a further analysis which presumably goes beyond pure statistics is needed to shed some light on it. We here do not attempt to provide any answers. We merely discuss some observations which become apparent upon considering our test results in the light of the Government Response Index (GRI). For reasons of brevity, we focus on the comparison of Germany with Italy and Spain in Figures \ref{fig:Germany:Italy} and \ref{fig:Germany:Spain}.

According to our test results in Figure \ref{fig:Germany:Spain}, there are significant differences between the trends in Germany and Spain during the onset of the epidemic up to about day $50$, with Spain having more new cases of infections than Germany on most days. After day $50$, the trends become quite similar and start to decrease at approximately the same rate. This may be due to the fact that Spain in general introduced more severe measures of lockdown than Germany (as can be seen upon inspecting the GRI in panel (c) of Figure \ref{fig:Germany:Spain}), which may have helped to battle the spread of infection. However, a much more thorough analysis is of course needed to find out whether this is indeed the case or whether other factors were mainly responsible.

Turning to the comparison of Germany and Italy, we found that the German trend drops down significantly faster than the Italian one after approximately day $40$. Interestingly, the GRI of Italy almost always lies above that of Germany. Hence, even though Italy has in general taken more severe and restrictive measures against the virus than Germany, it appears that the virus could be contained better in Germany (in the sense that the trend of daily new cases went down significantly faster in Germany than in Italy). This suggests that there are indeed important factors besides the level of government response to the pandemic which substantially influence the trend of new COVID-19 cases.

This brief discussion already indicates that it is extremely difficult to determine the exact causes of the differences in epidemic trends across countries. Since even similar countries such as those in our sample differ in a variety of aspects that are relevant for the spread of the virus, it is very challenging to pin down these causes. One issue that is often discussed in the context of cross-country comparisons are country-specific strategies to test for the coronavirus. The argument is that differences between epidemic trends may be spuriously produced by country-specific test procedures.

Even though we can of course not fully exclude this possibility, our test results are presumably not driven by different test regimes in the countries under consideration. To see this, we consider again the comparison of Germany and Italy: The test regimes in these two countries are arguably quite different. Germany is often cited as the country that employed early, widespread testing with more than $100\,000$ tests per week even in the beginning of the pandemic \citep{Cohen2020}, while testing in Italy became widespread only in the late stages of the pandemic. 
Nevertheless, visual inspection of the raw and smoothed data in panels (a) and (b) of Figure \ref{fig:Germany:Italy} suggest that the underlying time trends are very similar up to day $36$. This is confirmed by our multiscale test which does not find any significant differences before that day. Hence, the different test regimes in Germany and Italy towards the beginning of the pandemic do not appear to have an overly strong effect and to produce spurious differences between the time trends. This suggests that the differences detected by our multiscale test indeed reflect differences in the way the virus spread in Germany and Italy rather than being mere artefacts of different test regimes.

\newpage
\section*{A \hspace{0.2cm} Appendix}
\def\theequation{A.\arabic{equation}}
\setcounter{equation}{0}

\enlargethispage{0.1cm}

In what follows, we state and prove the main theoretical results on the multiscale test developed in Section \ref{sec:test}. Throughout the Appendix, we let $C$ be a generic positive constant that may take a different value on each occurrence. Unless stated differently, $C$ depends neither on the time series length $T$ nor on the dimension $p$ of the test problem. We further use the symbols $h_{\min}:= \min_{1 \le k \le K} h_k$ and $h_{\max} := \max_{1 \le k \le K} h_k$ to denote the smallest and largest interval length in the family $\intervals$.

\begin{theoremA}\label{theo1}
Let \ref{C2} and \ref{C1} be satisfied. Moreover, assume that (i) $h_{\max} = o(1/\log T)$, (ii) $h_{\min} \ge CT^{-b}$ for some $b \in (0,1)$, and (iii) $p = O(T^{(\theta/2)(1-b)-(1+\delta)})$ for some small $\delta > 0$. Then for any given $\alpha \in (0,1)$,
\[ \textnormal{FWER}(\alpha) := \pr \Big( \exists (i,j,k) \in \indexset_0: |\hat{\psi}_{ijk,T}| > c_{T,\textnormal{Gauss}}(\alpha,h_k) \Big) \le \alpha + o(1), \]
where $\indexset_0 \subseteq \indexset$ is the set of all $(i,j,k) \in \indexset$ for which $H_0^{(ijk)}$ holds true. 
\end{theoremA}

\noindent According to Theorem \ref{theo1}, the multiscale test asymptotically controls the FWER at level $\alpha$ under conditions \ref{C2}--\ref{C1} and the restrictions (i)--(iii) on $h_{\min}$, $h_{\max}$ and $p$. 
Restriction (i) allows the maximal interval length $h_{\max}$ to converge to zero very slowly, which means that $h_{\max}$ can be picked very large in practice. According to restriction (ii), the minimal interval length $h_{\min}$ can be chosen to go to zero as any polynomial $T^{-b}$ with some $b \in (0,1)$. Restriction (iii) allows the dimension $p$ of the test problem to grow polynomially in $T$. Specifically, $p$ may grow at most as the polynomial $T^{\gamma}$ with $\gamma = (\theta/2)(1-b)-(1+\delta)$. As one can see, the exponent $\gamma$ depends on the number of error moments $\theta$ defined in \ref{C1} and the parameter $b$ that specifies the minimal interval length $h_{\min}$. In particular, for any given $b \in (0,1)$, the exponent $\gamma$ gets larger as $\theta$ increases. Hence, the larger the number of error moments $\theta$, the faster $p$ may grow in comparison to $T$. In the extreme case where all error moments exist, that is, where $\theta$ can be made as large as desired, $p$ may grow as any polynomial of $T$, no matter how we pick $b \in (0,1)$. Thus, if the error terms have sufficiently many moments, the dimension $p$ can be extremely large in comparison to $T$ and the minimal interval length $h_{\min}$ can be chosen very small. 

The following corollary is an immediate consequence of Theorem \ref{theo1}. It provides the theoretical justification needed to make simultaneous confidence statements of the form \eqref{eq:confidencestatement1}--\eqref{eq:confidencestatement3}.

\begin{corollaryA}\label{corollary1}
Under the conditions of Theorem \ref{theo1}, 
\[ \pr\Big( \forall (i,j,k) \in \indexset: \text{ If } |\hat{\psi}_{ijk,T}| > c_{T,\textnormal{Gauss}}(\alpha,h_k), \text{ then } (i,j,k) \notin \indexset_0 \Big) \ge 1 - \alpha + o(1) \]
for any given $\alpha \in (0,1)$.   
\end{corollaryA}

\begin{proof}[\textnormal{\textbf{Proof of Theorem \ref{theo1}.}}] The proof proceeds in several steps. 
\begin{enumerate}[label=\textit{Step \arabic*.}, leftmargin=0cm, itemindent=1.45cm]

\item Let $\hat{\Psi}_T = \max_{(i,j,k) \in \indexset} a_k (|\hat{\psi}_{ijk,T}^0| - b_k)$ with $\hat{\psi}_{ijk,T}^0$ introduced in \eqref{eq:stat0} and define $\Psi_T = \max_{(i,j,k) \in \indexset} a_k (|\psi_{ijk,T}^0| - b_k)$ with 
\[ \psi_{ijk,T}^0 = \frac{1}{\sqrt{2Th_k}} \sum\limits_{t=1}^T \ind\Big(\frac{t}{T} \in \mathcal{I}_k\Big) (\eta_{it} - \eta_{jt}). \]
To start with, we prove that  
\begin{equation}\label{eq:approxerror1}
\big| \hat{\Psi}_T - \Psi_T \big| = o_p(r_T),
\end{equation}
where $\{r_T\}$ is any null sequence that converges more slowly to zero than $\rho_T = \sqrt{\log T} \{ \log p/\sqrt{Th_{\min}} + h_{\max} \sqrt{\log p} \}$, that is, $\rho_T/r_T \rightarrow 0$ as $T \rightarrow \infty$. Since the proof of \eqref{eq:approxerror1} is rather technical and lengthy, the details are provided in the Supplementary Material.

\item We next prove that 
\begin{equation}\label{eq:kolmogorov-distance}
\sup_{q \in \reals} \Big| \pr \big( \Psi_T \le q \big) - \pr \big( \Phi_T \le q \big) \Big| = o(1).
\end{equation}
To do so, we rewrite the statistics $\Psi_T$ and $\Phi_T$ as follows: Define 
\begin{equation*}
V^{(ijk)}_t = V^{(ijk)}_{t,T} := \sqrt{\frac{T}{2Th_k}} \ind\Big(\frac{t}{T} \in \mathcal{I}_k\Big) (\eta_{it} - \eta_{jt})
\end{equation*}
for $(i,j,k) \in \indexset$ and let $\boldsymbol{V}_t = (V_t^{(ijk)}: (i,j,k) \in \indexset)$ be the $p$-dimensional random vector with the entries $V_t^{(ijk)}$. With this notation, we get that $\psi_{ijk,T}^0 = T^{-1/2} \sum_{t=1}^T V^{(ijk)}_t$ and thus 
\begin{align*}
\Psi_T 
 & = \max_{(i,j,k) \in \indexset}  a_k \big( |\psi_{ijk,T}^0| - b_k \big) \\
 & = \max_{(i,j,k) \in \indexset} a_k \Big\{ \Big|\frac{1}{\sqrt{T}} \sum_{t=1}^T V^{(ijk)}_t\Big| - b_k \Big\}.
\end{align*} 
Analogously, we define 
\begin{equation*}
W^{(ijk)}_t = W^{(ijk)}_{t,T} := \sqrt{\frac{T}{2Th_k}} \ind\Big(\frac{t}{T} \in \mathcal{I}_k\Big) (Z_{it} - Z_{jt})
\end{equation*}
with $Z_{it}$ i.i.d.\ standard normal and let $\boldsymbol{W}_t = (W_t^{(ijk)}: (i,j,k) \in \indexset)$. The vector $\boldsymbol{W}_t$ is a Gaussian version of $\boldsymbol{V}_t$ with the same mean and variance. In particular, $\ex[\boldsymbol{W}_t] = \ex[\boldsymbol{V}_t] = 0$ and $\ex[\boldsymbol{W}_t \boldsymbol{W}_t^\top] = \ex[\boldsymbol{V}_t \boldsymbol{V}_t^\top]$. Similarly as before, we can write $\phi_{ijk,T} = T^{-1/2} \sum_{t=1}^T W^{(ijk)}_t$ and  
\begin{align*}
\Phi_T 
 & = \max_{(i,j,k) \in \indexset} a_k \big( |\phi_{ijk,T}| - b_k \big) \\
 & = \max_{(i,j,k) \in \indexset} a_k \Big\{ \Big|\frac{1}{\sqrt{T}} \sum_{t=1}^T W^{(ijk)}_t\Big| - b_k \Big\}.
\end{align*} 
For any $q \in \reals$, it holds that
\begin{align*}
\pr \big( \Psi_T \le q \big) 
 & = \pr \Big( \max_{(i,j,k) \in \indexset} a_k \Big\{ \Big|\frac{1}{\sqrt{T}} \sum_{t=1}^T V^{(ijk)}_t\Big| - b_k \Big\} \le q \Big) \\
 & = \pr \Big( \Big|\frac{1}{\sqrt{T}} \sum_{t=1}^T V^{(ijk)}_t\Big| \le c_{ijk}(q) \text{ for all } (i,j,k) \in \indexset \Big) \\
 & = \pr \Big( \Big|\frac{1}{\sqrt{T}} \sum_{t=1}^T \boldsymbol{V}_t\Big| \le \boldsymbol{c}(q) \Big),
\end{align*} 
where $\boldsymbol{c}(q) = (c_{ijk}(q): (i,j,k) \in \indexset)$ is the $\reals^p$-vector with the entries $c_{ijk}(q) = q/a_k + b_k$, we use the notation $|v| = (|v_1|,\ldots,|v_p|)^\top$ for vectors $v \in \reals^p$ and the inequality $v \le w$ is to be understood componentwise for $v,w \in \reals^p$. Analogously, we have  
\[ \pr \big( \Phi_T \le q \big) = \pr \Big( \Big|\frac{1}{\sqrt{T}} \sum_{t=1}^T \boldsymbol{W}_t\Big| \le \boldsymbol{c}(q) \Big). \]
With this notation at hand, we can make use of Proposition 2.1 from \cite{Chernozhukov2017}. In our context, this proposition can be stated as follows: 
\begin{propA}\label{prop:Chernozhukov}
Assume that 
\begin{enumerate}[label=(\alph*),leftmargin=0.7cm]
\item $T^{-1} \sum_{t=1}^T \ex (V^{(ijk)}_t)^2 \ge \delta > 0$ for all $(i,j,k) \in \indexset$.
\item $T^{-1} \sum_{t=1}^T \ex[ |V^{(ijk)}_t|^{2+r} ] \le B_T^r$ for all $(i,j,k) \in \indexset$ and $r=1,2$, where $B_T \ge 1$ are constants that may tend to infinity as $T \rightarrow \infty$.  
\item $\ex[ \{ \max_{(i,j,k) \in \indexset} |V^{(ijk)}_t| / B_T \}^\theta ] \le 2$ for all $t$ and some $\theta > 4$.  
\end{enumerate}
Then  
\begin{align}
\sup_{\boldsymbol{c} \in \reals^p} \Big| \pr \Big( \Big|\frac{1}{\sqrt{T}} \sum_{t=1}^T \boldsymbol{V}_t\Big| \le \boldsymbol{c} \Big) & - \pr \Big( \Big|\frac{1}{\sqrt{T}} \sum_{t=1}^T \boldsymbol{W}_t\Big| \le \boldsymbol{c} \Big) \Big| \nonumber \\ & \le C \Big\{ \Big( \frac{B_T^2 \log^7(pT)}{T} \Big)^{1/6} + \Big( \frac{B_T^2 \log^3(pT)}{T^{1-2/\theta}} \Big)^{1/3} \Big\}, \label{eq:Chernozhukov}
\end{align}
where $C$ depends only on $\delta$ and $\theta$. 
\end{propA}
It is straightforward to verify that assumptions (a)--(c) are satisfied under the conditions of Theorem \ref{theo1} for sufficiently large $T$, where $B_T$ can be chosen as $B_T = C p^{1/\theta} h_{\min}^{-1/2}$ with $C$ sufficiently large. Moreover, it can be shown that the right-hand side of \eqref{eq:Chernozhukov} is $o(1)$ for this choice of $B_T$. Hence, Proposition \ref{prop:Chernozhukov} yields that 
\[ \sup_{\boldsymbol{c} \in \reals^p} \Big| \pr \Big( \Big|\frac{1}{\sqrt{T}} \sum_{t=1}^T \boldsymbol{V}_t\Big| \le \boldsymbol{c} \Big) - \pr \Big( \Big|\frac{1}{\sqrt{T}} \sum_{t=1}^T \boldsymbol{W}_t\Big| \le \boldsymbol{c} \Big) \Big| = o(1), \]
which in turn implies \eqref{eq:kolmogorov-distance}.

\item With the help of \eqref{eq:approxerror1} and \eqref{eq:kolmogorov-distance}, we now show that 
\begin{equation}\label{eq:kolmogorov-distance-hat}
\sup_{q \in \reals} \Big| \pr \big( \hat{\Psi}_T \le q \big) - \pr \big( \Phi_T \le q \big) \Big| = o(1).
\end{equation}
To start with, the above supremum can be bounded by 
\begin{align}
 & \sup_{q \in \reals} \Big| \pr \big( \hat{\Psi}_T \le q \big) - \pr \big( \Phi_T \le q \big) \Big| \nonumber \\
 & = \sup_{q \in \reals} \Big| \pr \Big( \Psi_T \le q + \big\{ \Psi_T - \hat{\Psi}_T \big\} \Big) - \pr \big( \Phi_T \le q \big) \Big| \nonumber \\
 & \le \sup_{q \in \reals} \max \Big\{ \Big| \pr \Big( \Psi_T \le q + \big| \Psi_T - \hat{\Psi}_T \big| \Big) - \pr \big( \Phi_T \le q \big) \Big|, \nonumber \\
 & \phantom{\le \sup_{q \in \reals} \max \Big\{ \ } \Big| \pr \Big( \Psi_T \le q - \big| \Psi_T - \hat{\Psi}_T \big| \Big) - \pr \big( \Phi_T \le q \big) \Big| \Big\} \nonumber \\
 & \le \sup_{q \in \reals} \max \Big\{ \Big| \pr \Big( \Psi_T \le q + r_T \Big) - \pr \big( \Phi_T \le q \big) \Big| + \pr \Big( \big| \Psi_T - \hat{\Psi}_T \big| > r_T \Big), \nonumber \\
 & \phantom{\le \sup_{q \in \reals} \max \Big\{ \ } \Big| \pr \Big( \Psi_T \le q - r_T \Big) - \pr \big( \Phi_T \le q \big) \Big| + \pr \Big( \big| \Psi_T - \hat{\Psi}_T \big| > r_T \Big) \Big\} \nonumber \\
 & \le \max_{\ell=0,1} \, \sup_{q \in \reals} \Big| \pr \Big( \Psi_T \le q + (-1)^\ell r_T \Big) - \pr \big( \Phi_T \le q \big) \Big| + \pr \Big( \big| \Psi_T - \hat{\Psi}_T \big| > r_T \Big) \nonumber \\
 & = \max_{\ell=0,1} \, \sup_{q \in \reals} \Big| \pr \Big( \Psi_T \le q + (-1)^\ell r_T \Big) - \pr \big( \Phi_T \le q \big) \Big| + o(1), \label{eq:step3:a}
\end{align}
where the last line is by \eqref{eq:approxerror1}. Moreover, for $\ell=0,1$, 
\begin{align}
 & \sup_{q \in \reals} \Big| \pr \Big( \Psi_T \le q + (-1)^\ell r_T \Big) - \pr \big( \Phi_T \le q \big) \Big| \nonumber \\
 & \le \sup_{q \in \reals} \Big| \pr \Big( \Psi_T \le q + (-1)^\ell r_T \Big) - \pr \Big( \Phi_T \le q + (-1)^\ell r_T \Big) \Big| \nonumber \\
 & \quad + \sup_{q \in \reals} \Big| \pr \Big( \Phi_T \le q + (-1)^\ell r_T \Big) - \pr \big( \Phi_T \le q \big) \Big| \nonumber \\
 & = \sup_{q \in \reals} \Big| \pr \Big( \Phi_T \le q + (-1)^\ell r_T \Big) - \pr \big( \Phi_T \le q \big) \Big| + o(1), \label{eq:step3:b}
\end{align}
the last line following from \eqref{eq:kolmogorov-distance}. Finally, by Nazarov's inequality (cp.\ \citeauthor{Nazarov2003}, \citeyear{Nazarov2003} and Lemma A.1 in \citeauthor{Chernozhukov2017}, \citeyear{Chernozhukov2017}), we have that for $\ell = 0,1$,   
\begin{align} 
 & \sup_{q \in \reals} \Big| \pr \Big( \Phi_T \le q + (-1)^\ell r_T \Big) - \pr \big( \Phi_T \le q \big) \Big| \nonumber \\
 & = \sup_{q \in \reals} \Big| \pr \Big( \Big|\frac{1}{\sqrt{T}} \sum_{t=1}^T \boldsymbol{W}_t\Big| \le \boldsymbol{c}(q + (-1)^\ell r_T) \Big) - \pr \Big( \Big|\frac{1}{\sqrt{T}} \sum_{t=1}^T \boldsymbol{W}_t\Big| \le \boldsymbol{c}(q) \Big) \Big| \nonumber \\
 & \le C r_T \sqrt{\log(2 p)}, \label{eq:step3:c}
\end{align}
where $C$ is a constant that depends only on the parameter $\delta$ defined in condition (a) of Proposition \ref{prop:Chernozhukov}. Inserting \eqref{eq:step3:b} and \eqref{eq:step3:c} into equation \eqref{eq:step3:a} completes the proof of \eqref{eq:kolmogorov-distance-hat}.

\item By definition of the quantile $q_{T,\text{Gauss}}(\alpha)$, it holds that $\pr(\Phi_T \le q_{T,\text{Gauss}}(\alpha)) \ge 1-\alpha$. As shown in the Supplementary Material, we even have that  
\begin{equation}\label{eq:quant-exact}
\pr(\Phi_T \le q_{T,\text{Gauss}}(\alpha)) = 1-\alpha
\end {equation} 
for any $\alpha \in (0,1)$. From this and \eqref{eq:kolmogorov-distance-hat}, it immediately follows that  
\begin{equation}\label{eq:probbound-Psihat}
\pr \big( \hat{\Psi}_T \le q_{T,\text{Gauss}}(\alpha) \big) = 1 - \alpha + o(1), 
\end{equation}
which in turn implies that 
\begin{align*}
\text{FWER}(\alpha)
 & = \pr \Big( \exists (i,j,k) \in \indexset_0: |\hat{\psi}_{ijk,T}| > c_{T,\text{Gauss}}(\alpha,h_k) \Big) \\
 & = \pr \Big( \max_{(i,j,k) \in \indexset_0} a_k \big( |\hat{\psi}_{ijk,T}| - b_k \big) > q_{T,\text{Gauss}}(\alpha) \Big) \\
 & = \pr \Big( \max_{(i,j,k) \in \indexset_0} a_k \big( |\hat{\psi}_{ijk,T}^0| - b_k \big) > q_{T,\text{Gauss}}(\alpha) \Big) \\
 & \le \pr \Big( \max_{(i,j,k) \in \indexset} a_k \big( |\hat{\psi}_{ijk,T}^0| - b_k \big) > q_{T,\text{Gauss}}(\alpha) \Big) \\
 & = \pr \big( \hat{\Psi}_T > q_{T,\text{Gauss}}(\alpha) \big) = \alpha + o(1).
\end{align*}
This completes the proof of Theorem \ref{theo1}. \qedhere

\end{enumerate}
\end{proof}

\begin{proof}[\textnormal{\textbf{Proof of Corollary \ref{corollary1}.}}]
By Theorem \ref{theo1}, 
\begin{align*}
1 - \alpha + o(1) 
 & \le 1 - \textnormal{FWER}(\alpha) \\
 & = \pr \Big( \nexists (i,j,k) \in \indexset_0: |\hat{\psi}_{ijk,T}| > c_{T,\textnormal{Gauss}}(\alpha,h_k) \Big) \\
 & = \pr\Big( \forall (i,j,k) \in \indexset: \text{ If } |\hat{\psi}_{ijk,T}| > c_{T,\textnormal{Gauss}}(\alpha,h_k), \text{ then } (i,j,k) \notin \indexset_0 \Big),
\end{align*}
which gives the statement of Corollary \ref{corollary1}.
\end{proof}

\section*{Acknowledgements}

The authors gratefully acknowledge financial support by the Deutsche Forschungsgemeinschaft (DFG, German Research Foundation), grant number 430668955.

\bibliographystyle{ims}
{\small
\setlength{\bibsep}{0.35em}
\bibliography{bibliography}}

\newpage
\def\thesection{\Alph{section}}
\setcounter{section}{18}
\section{Supplementary Material}


\def\theequation{S.\arabic{equation}}
\setcounter{equation}{0}
\def\thefigure{S.\arabic{figure}}
\setcounter{figure}{0}
\def\thetable{S.\arabic{table}}
\setcounter{table}{0}
\renewcommand{\baselinestretch}{1.2}\normalsize

\subsection{Technical details}

In what follows, we provide the technical details omitted in the Appendix. To start with, we prove the following auxiliary lemma.
\begin{lemmaS}\label{lemmaS1}
Under the conditions of Theorem \ref{theo1}, it holds that 
\[ \big| \hat{\sigma}^2 - \sigma^2 \big| = O_p\Big( \sqrt{\frac{\log p}{T}} \Big). \]
\end{lemmaS}

\begin{proof}[\textnormal{\textbf{Proof of Lemma \ref{lemmaS1}.}}]
By definition, $\hat{\sigma}^2 = |\countries|^{-1} \sum_{i \in \countries} \hat{\sigma}_i^2$ and $\hat{\sigma}_i^2 = \{\sum_{t=2}^T (\X_{it}-\X_{it-1})^2\}\{2 \sum_{t=1}^T \X_{it}\}^{-1}$. It holds that 
\begin{equation}\label{eq:approxerror:claim3:expansion}
\frac{1}{T} \sum\limits_{t=2}^T (X_{it} - X_{it-1})^2 = \frac{\sigma^2}{T} \sum\limits_{t=2}^T \lambda_i\Big(\frac{t}{T}\Big) (\eta_{it} - \eta_{it-1})^2 + \big\{R_{i,T}^{(1)} + \ldots + R_{i,T}^{(5)}\big\},
\end{equation}
where
\begin{align*}
R_{i,T}^{(1)} & = \frac{2\sigma}{T} \sum\limits_{t=2}^T \Big( \lambda_i\Big(\frac{t}{T}\Big) - \lambda_i\Big(\frac{t-1}{T}\Big) \Big) \sqrt{\lambda_i\Big(\frac{t}{T}\Big)} (\eta_{it} - \eta_{it-1}) \\
R_{i,T}^{(2)} & = \frac{2\sigma^2}{T} \sum\limits_{t=2}^T \Big( \sqrt{\lambda_i\Big(\frac{t}{T}\Big)} - \sqrt{\lambda_i\Big(\frac{t-1}{T}\Big)} \Big) \sqrt{\lambda_i\Big(\frac{t}{T}\Big)} \eta_{it-1} (\eta_{it} - \eta_{it-1}) \\
R_{i,T}^{(3)} & = \frac{1}{T} \sum\limits_{t=2}^T \Big( \lambda_i\Big(\frac{t}{T}\Big) - \lambda_i\Big(\frac{t-1}{T}\Big) \Big)^2 \\
R_{i,T}^{(4)} & = \frac{2 \sigma}{T} \sum\limits_{t=2}^T \Big( \lambda_i\Big(\frac{t}{T}\Big) - \lambda_i\Big(\frac{t-1}{T}\Big) \Big) \Big( \sqrt{\lambda_i\Big(\frac{t}{T}\Big)} - \sqrt{\lambda_i\Big(\frac{t-1}{T}\Big)} \Big) \eta_{it-1} \\
R_{i,T}^{(5)} & = \frac{\sigma^2}{T} \sum\limits_{t=2}^T  \Big( \sqrt{\lambda_i\Big(\frac{t}{T}\Big)} - \sqrt{\lambda_i\Big(\frac{t-1}{T}\Big)} \Big)^2 \eta_{it-1}^2.
\end{align*} 
With the help of an exponential inequality and standard arguments, it can be shown that 
\[ \max_{i \in \countries} \Big| \frac{1}{T} \sum_{t=2}^T w_i\Big(\frac{t}{T}\Big) \big\{ g(\eta_{it},\eta_{it-1}) - \ex g(\eta_{it},\eta_{it-1}) \big\} \Big| = O_p\Big( \sqrt{\frac{\log p}{T}} \Big), \]
where we let $g(x,y) = x$, $g(x,y) = y$, $g(x,y) = |x|$, $g(x,y) = |y|$, $g(x,y) = x^2$, $g(x,y) = y^2$ or $g(x,y) = xy$, and $w_i(t/T)$ are deterministic weights with the property that $|w_i(t/T)| \le w_{\max} < \infty$ for all $i$, $t$ and $T$ and some positive constant $w_{\max}$. Using this uniform convergence result along with conditions \ref{C2} and \ref{C1}, we obtain that 
\[ \max_{i \in \countries} \Big| \frac{1}{T} \sum\limits_{t=2}^T \lambda_i\Big(\frac{t}{T}\Big) (\eta_{it} - \eta_{it-1})^2 - \frac{2}{T} \sum\limits_{t=1}^T \lambda_i\Big(\frac{t}{T}\Big) \Big| = O_p\Big( \sqrt{\frac{\log p}{T}} \Big) \]
and 
\[ \max_{1 \le \ell \le 5} \max_{i \in \countries} |R_{i,T}^{(\ell)}| = O_p(T^{-1}). \]
Applying these two statements to \eqref{eq:approxerror:claim3:expansion}, we can infer that
\begin{equation}\label{eq:approxerror:claim2:1}
\max_{i \in \countries} \Big| \frac{1}{T} \sum\limits_{t=2}^T (X_{it} - X_{it-1})^2 - \frac{2\sigma^2}{T} \sum\limits_{t=1}^T \lambda_i\Big(\frac{t}{T}\Big) \Big| = O_p\Big( \sqrt{\frac{\log p}{T}} \Big). 
\end{equation}
By similar but simpler arguments, we additionally get that 
\begin{equation}\label{eq:approxerror:claim2:2}
\max_{i \in \countries} \Big| \frac{1}{T} \sum\limits_{t=1}^T X_{it} - \frac{1}{T} \sum\limits_{t=1}^T \lambda_i\Big(\frac{t}{T}\Big) \Big| = O_p\Big( \sqrt{\frac{\log p}{T}} \Big). 
\end{equation}
From \eqref{eq:approxerror:claim2:1} and \eqref{eq:approxerror:claim2:2}, it follows that $\max_{i \in \countries} |\hat{\sigma}_i^2 - \sigma^2| = O_p(\sqrt{\log p / T})$, which in turn implies that $|\hat{\sigma}^2 - \sigma^2| = O_p(\sqrt{\log p / T})$ as well. 
\end{proof}

\begin{proof}[\textnormal{\textbf{Proof of (\ref{eq:approxerror1}).}}] 
Since 
\begin{align*}
\big| \hat{\Psi}_T - \Psi_T \big| 
 & \le \max_{(i,j,k) \in \indexset} a_k \big| \hat{\psi}_{ijk,T}^0 - \psi_{ijk,T}^0 \big| \\
 & \le \max_{1 \le k \le K} a_k \max_{(i,j,k) \in \indexset} \big| \hat{\psi}_{ijk,T}^0 - \psi_{ijk,T}^0 \big| \\
 & \le C \sqrt{\log T} \max_{(i,j,k) \in \indexset} \big| \hat{\psi}_{ijk,T}^0 - \psi_{ijk,T}^0 \big|, 
\end{align*}
it suffices to prove that 
\begin{equation}\label{eq:approxerror2}
\max_{(i,j,k) \in \indexset} \big| \hat{\psi}_{ijk,T}^0 - \psi_{ijk,T}^0 \big| = o_p\Big(\frac{r_T}{\sqrt{\log T}}\Big).
\end{equation}
To start with, we reformulate $\hat{\psi}_{ijk,T}^0$ as
\[ \hat{\psi}_{ijk,T}^0 = \hat{\psi}_{ijk,T}^* + \Big( \frac{\sigma}{\hat{\sigma}} - 1 \Big) \hat{\psi}_{ijk,T}^*, \]
where 
\[ \hat{\psi}_{ijk,T}^* =  \frac{\sum\nolimits_{t=1}^T \ind(\frac{t}{T} \in \mathcal{I}_k) \overline{\lambda}_{ij}^{1/2}(\frac{t}{T}) (\eta_{it} - \eta_{jt})}{ \{ \sum\nolimits_{t=1}^T \ind(\frac{t}{T} \in \mathcal{I}_k) (\X_{it} + \X_{jt}) \}^{1/2}}. \]
With this notation, we can establish the bound 
\begin{align*}
\max_{(i,j,k) \in \indexset} \big| \hat{\psi}_{ijk,T}^0 - \psi_{ijk,T}^0 \big| 
 & \le \max_{(i,j,k) \in \indexset} \big| \hat{\psi}_{ijk,T}^* - \psi_{ijk,T}^0 \big| \\
 & \quad + \Big| \frac{\sigma}{\hat{\sigma}} - 1 \Big| \max_{(i,j,k) \in \indexset} \big| \hat{\psi}_{ijk,T}^* - \psi_{ijk,T}^0 \big| \\
 & \quad + \Big| \frac{\sigma}{\hat{\sigma}} - 1 \Big| \max_{(i,j,k) \in \indexset} \big|\psi_{ijk,T}^0\big|, 
\end{align*}
which shows that \eqref{eq:approxerror2} is implied by the three statements 
\begin{align} 
\max_{(i,j,k) \in \indexset} \big| \hat{\psi}_{ijk,T}^* - \psi_{ijk,T}^0 \big| & = O_p \Big( \frac{\log p}{\sqrt{T h_{\min}}} + h_{\max} \sqrt{\log p} \Big) \label{eq:approxerror:claim1} \\
\max_{(i,j,k) \in \indexset} \big|\psi_{ijk,T}^0\big| & = O_p\big(\sqrt{\log p}\big) \label{eq:approxerror:claim2} \\
\big| \hat{\sigma}^2 - \sigma^2 \big| & = O_p\Big( \sqrt{\frac{\log p}{T}} \Big). \label{eq:approxerror:claim3} 
\end{align}
Since \eqref{eq:approxerror:claim3} has already been verified in Lemma \ref{lemmaS1}, it remains to prove the statements \eqref{eq:approxerror:claim1} and \eqref{eq:approxerror:claim2}.

We start with the proof of \eqref{eq:approxerror:claim2}. Applying an exponential inequality along with standard arguments yields that 
\begin{equation}\label{eq:approxerror:var}
\max_{i \in \countries} \max_{1 \le k \le K} \Big| \frac{1}{\sqrt{Th_k}} \sum\limits_{t=1}^T \ind\Big(\frac{t}{T} \in \mathcal{I}_k\Big) w_i\Big(\frac{t}{T}\Big) \eta_{it} \Big| = O_p \big( \sqrt{\log p} \big),
\end{equation}
where $w_i(t/T)$ are general deterministic weights with the property that $|w_i(t/T)| \le w_{\max} < \infty$ for all $i$, $t$ and $T$ and some positive constant $w_{\max}$. This immediately implies \eqref{eq:approxerror:claim2}.

We next turn to the proof of \eqref{eq:approxerror:claim1}. As the functions $\lambda_i$ are uniformly Lipschitz continuous by \ref{C2}, it can be shown that 
\begin{equation}\label{eq:approxerror:bias}
\max_{i \in \countries} \max_{1 \le k \le K} \Big| \frac{1}{Th_k} \sum\limits_{t=1}^T \ind\Big(\frac{t}{T} \in \mathcal{I}_k\Big) \lambda_i \Big(\frac{t}{T}\Big) - \frac{1}{h_k} \int_{w \in \mathcal{I}_k} \lambda_i(w) dw \Big| \le \frac{C}{T h_{\min}}.
\end{equation}
From this, the uniform convergence result \eqref{eq:approxerror:var} and condition \ref{C2}, we can infer that 
\begin{align}
 \max_{(i,j,k) \in \indexset} \Big| & \frac{1}{Th_k} \sum\limits_{t=1}^T \ind\Big(\frac{t}{T} \in \mathcal{I}_k\Big) (\X_{it} + \X_{jt}) \nonumber \\* & - \frac{1}{h_k} \int_{w \in \mathcal{I}_k} \big\{ \lambda_i(w) + \lambda_j(w) \big\} dw \Big| = O_p \Big( \sqrt{\frac{\log p}{T h_{\min}}} \Big) \label{eq:approxerror:claim1:infer1} 
\end{align}
and
\begin{align}
 \max_{(i,j,k) \in \indexset} \Big| & \frac{1}{\sqrt{Th_k}} \sum\limits_{t=1}^T \ind\Big(\frac{t}{T} \in \mathcal{I}_k\Big) \overline{\lambda}_{ij}^{1/2}\Big(\frac{t}{T}\Big) (\eta_{it} - \eta_{jt}) \nonumber \\* & - \Big\{\frac{\int_{w \in \mathcal{I}_k} \overline{\lambda}_{ij}(w) dw}{h_k}\Big\}^{1/2} \frac{1}{\sqrt{Th_k}} \sum\limits_{t=1}^T \ind\Big(\frac{t}{T} \in \mathcal{I}_k\Big) (\eta_{it} - \eta_{jt}) \Big| \nonumber \\ & \hspace{6.5cm} = O_p \Big( h_{\max} \sqrt{\log p} \Big). \label{eq:approxerror:claim1:infer2}
\end{align}
The claim \eqref{eq:approxerror:claim1} follows from \eqref{eq:approxerror:claim1:infer1} and \eqref{eq:approxerror:claim1:infer2} along with straightforward calculations.
\end{proof}

\begin{proof}[\textnormal{\textbf{Proof of (\ref{eq:quant-exact}).}}]
The proof is by contradiction. Suppose that \eqref{eq:quant-exact} does not hold true, that is, $\pr(\Phi_T \le q_{T,\text{Gauss}}(\alpha)) = 1-\alpha + \xi$ for some $\xi > 0$. By Nazarov's inequality, 
\begin{equation*}
\pr \big(\Phi_T \le q_{T,\text{Gauss}}(\alpha)\big) - \pr\big(\Phi_T \le q_{T,\text{Gauss}}(\alpha) - \eta\big) \le C \eta \sqrt{\log (2p)} 
\end{equation*}
for any $\eta > 0$ with $C$ depending only on the parameter $\delta$ specified in condition (a) of Proposition \ref{prop:Chernozhukov}. Hence, 
\begin{align*}
\pr \big(\Phi_T \le q_{T,\text{Gauss}}(\alpha) - \eta\big) 
 & \ge \pr\big(\Phi_T \le q_{T,\text{Gauss}}(\alpha) \big) - C \eta \sqrt{\log(2p)} \\
 & = 1-\alpha + \xi - C \eta \sqrt{\log(2p)} > 1-\alpha
\end{align*}
for $\eta > 0$ sufficiently small. This contradicts the definition of the quantile $q_{T,\text{Gauss}}(\alpha)$ according to which $q_{T,\text{Gauss}}(\alpha) = \inf_{q \in \reals} \{ \pr(\Phi_T \le q) \ge 1-\alpha \}$. 
\end{proof}

\newpage
\subsection{Additional graphs for Section \ref{subsec:app}}\label{s:subsec:app}

Here, we provide the pairwise comparisons between Italy, France, Spain and the UK that were omitted in Section \ref{subsec:app}. The plots have the same format as Figures \ref{fig:Germany:Italy}--\ref{fig:Germany:UK}.

\begin{figure}[h!]
\begin{minipage}[t]{0.49\textwidth}
\includegraphics[width=\textwidth]{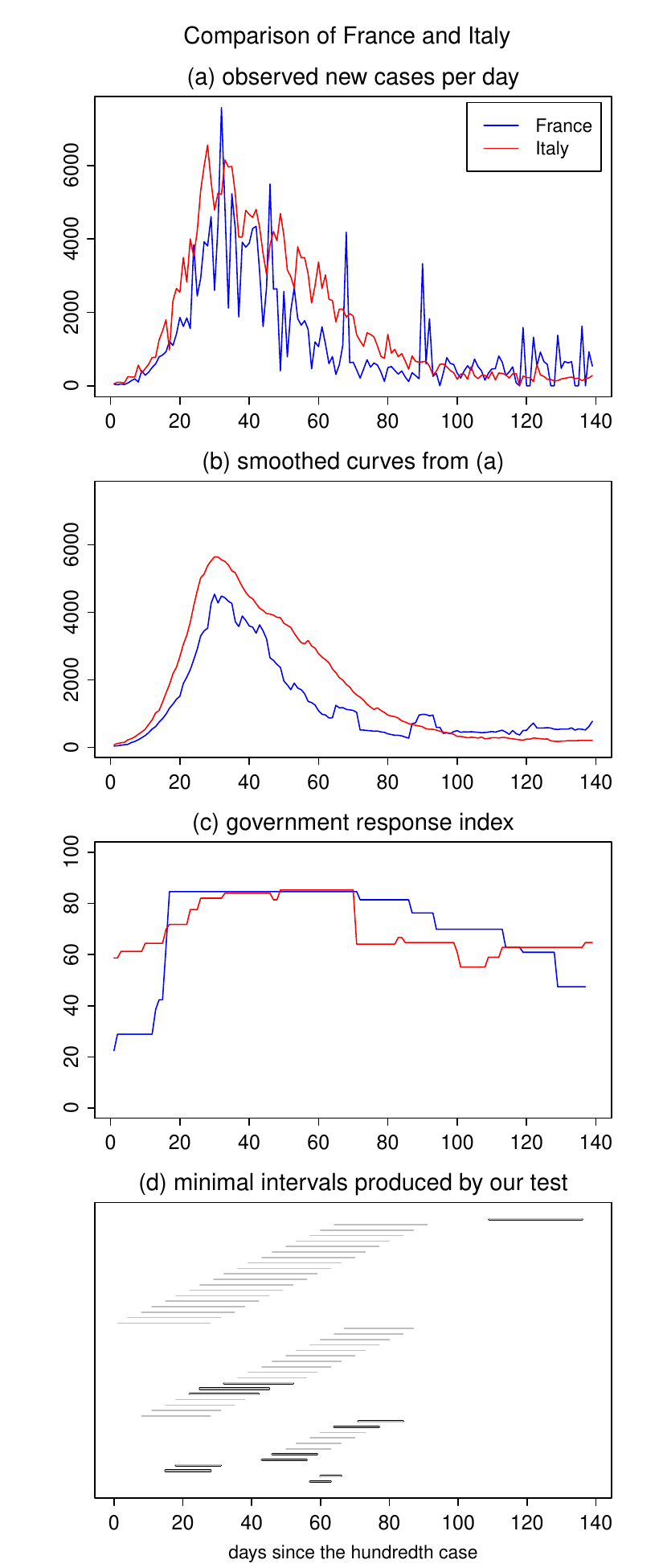}
\caption{Test results for the comparison of France and Italy.}
\end{minipage}
\hspace{0.25cm}
\begin{minipage}[t]{0.49\textwidth}
\includegraphics[width=\textwidth]{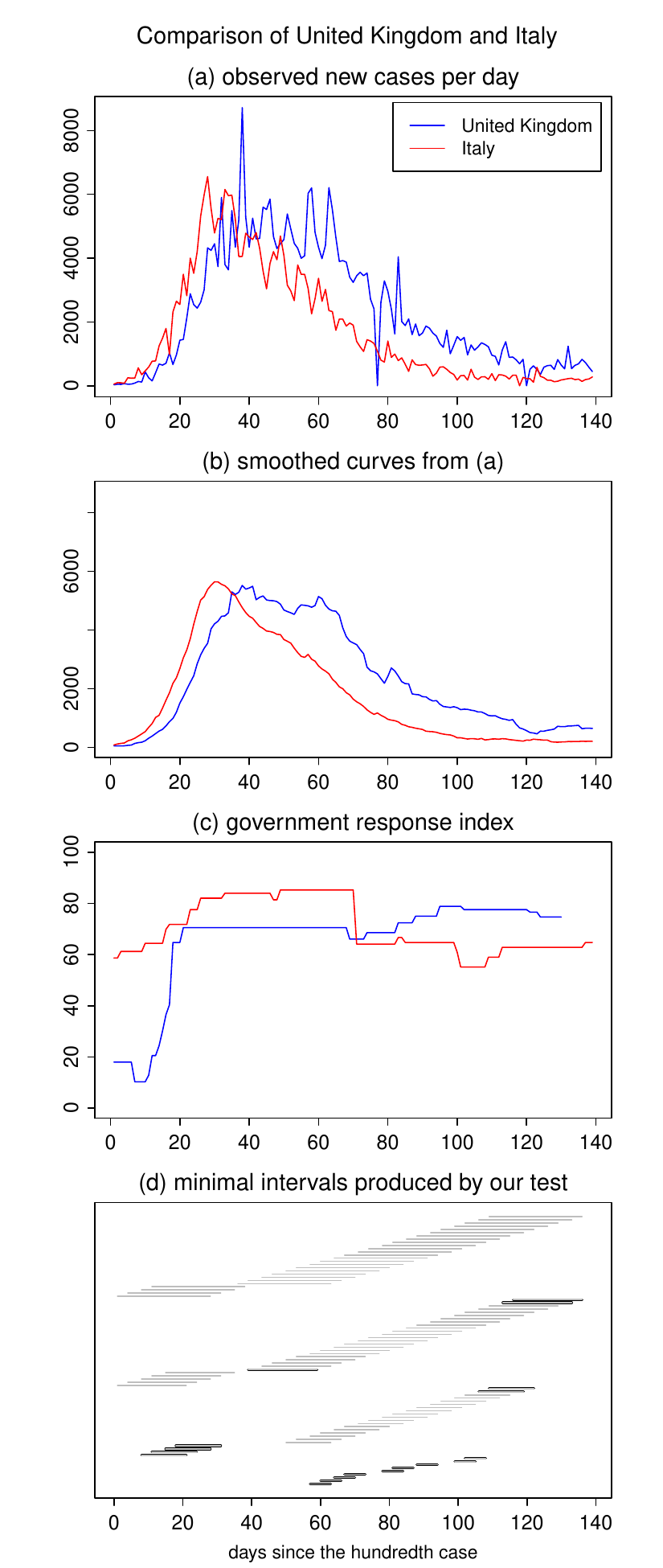}
\caption{Test results for the comparison of the UK and Italy.}
\end{minipage}
\end{figure}

\begin{figure}[p!]
\begin{minipage}[t]{0.49\textwidth}
\includegraphics[width=\textwidth]{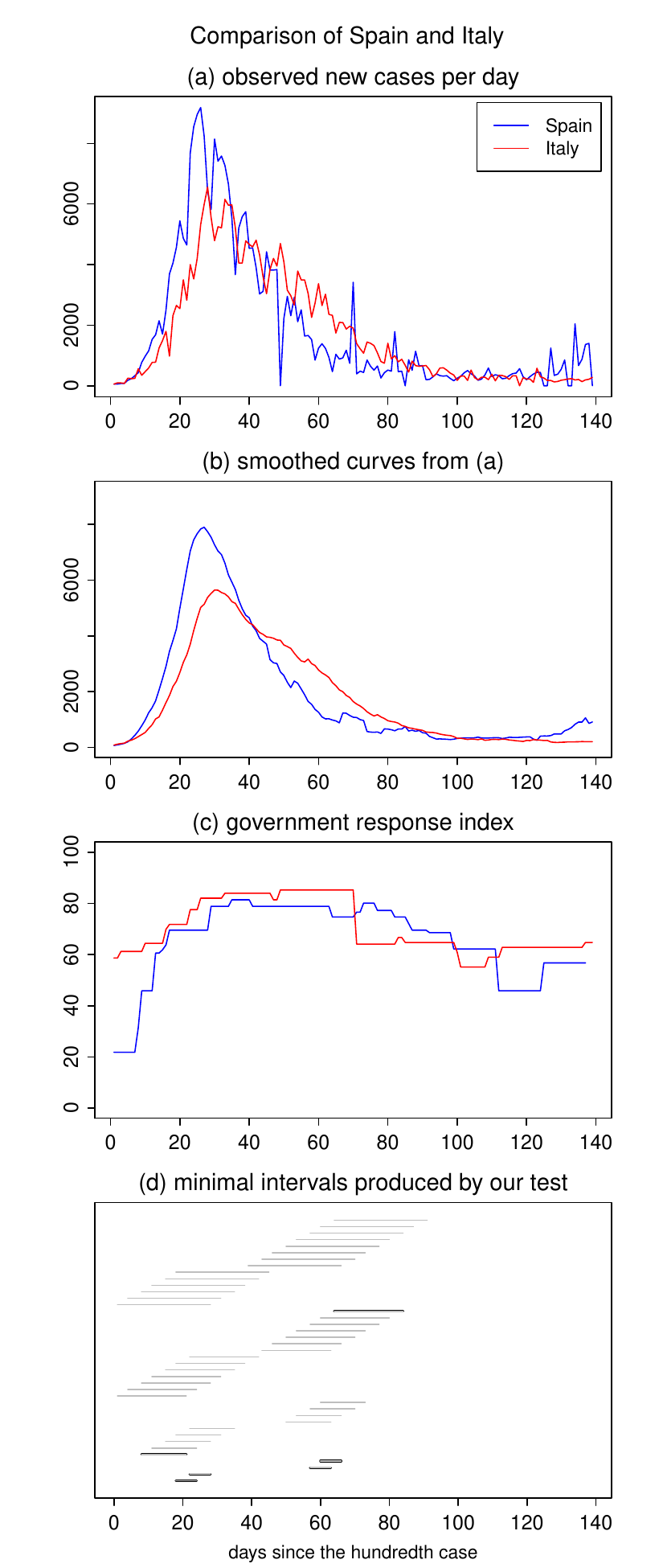}
\caption{Test results for the comparison of Spain and Italy.}
\end{minipage}
\hspace{0.25cm}
\begin{minipage}[t]{0.49\textwidth}
\includegraphics[width=\textwidth]{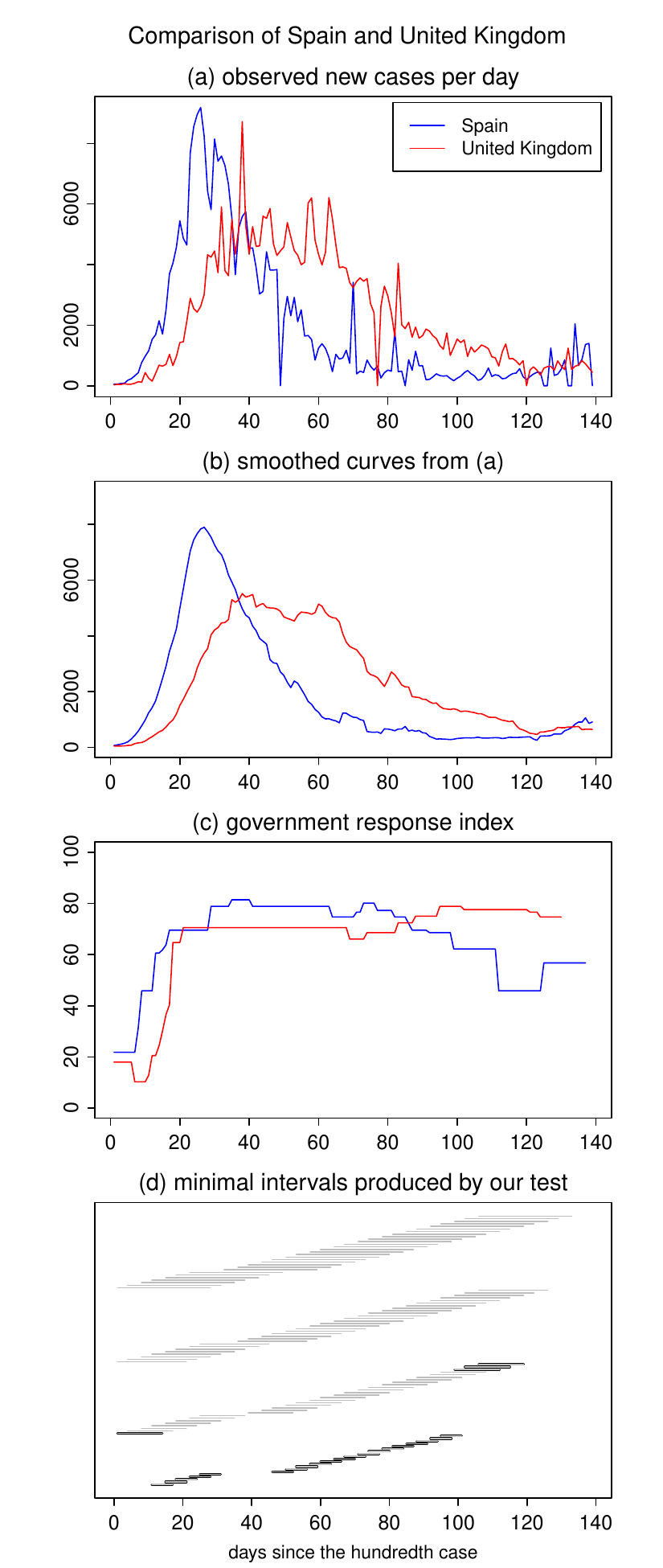}
\caption{Test results for the comparison of Spain and the UK.}
\end{minipage}
\end{figure}

\begin{figure}[p!]
\begin{minipage}[t]{0.49\textwidth}
\includegraphics[width=\textwidth]{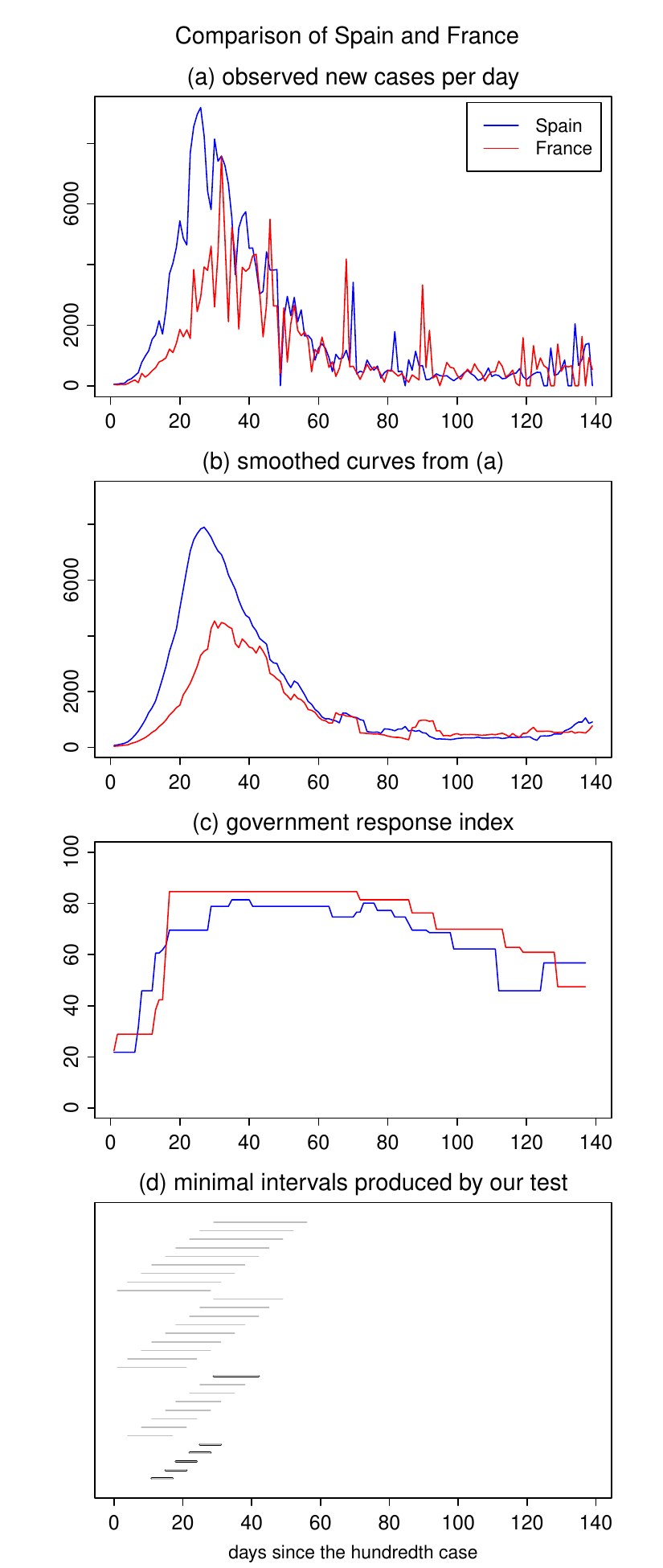}
\caption{Test results for the comparison of Spain and France.}
\end{minipage}
\hspace{0.25cm}
\begin{minipage}[t]{0.49\textwidth}
\includegraphics[width=\textwidth]{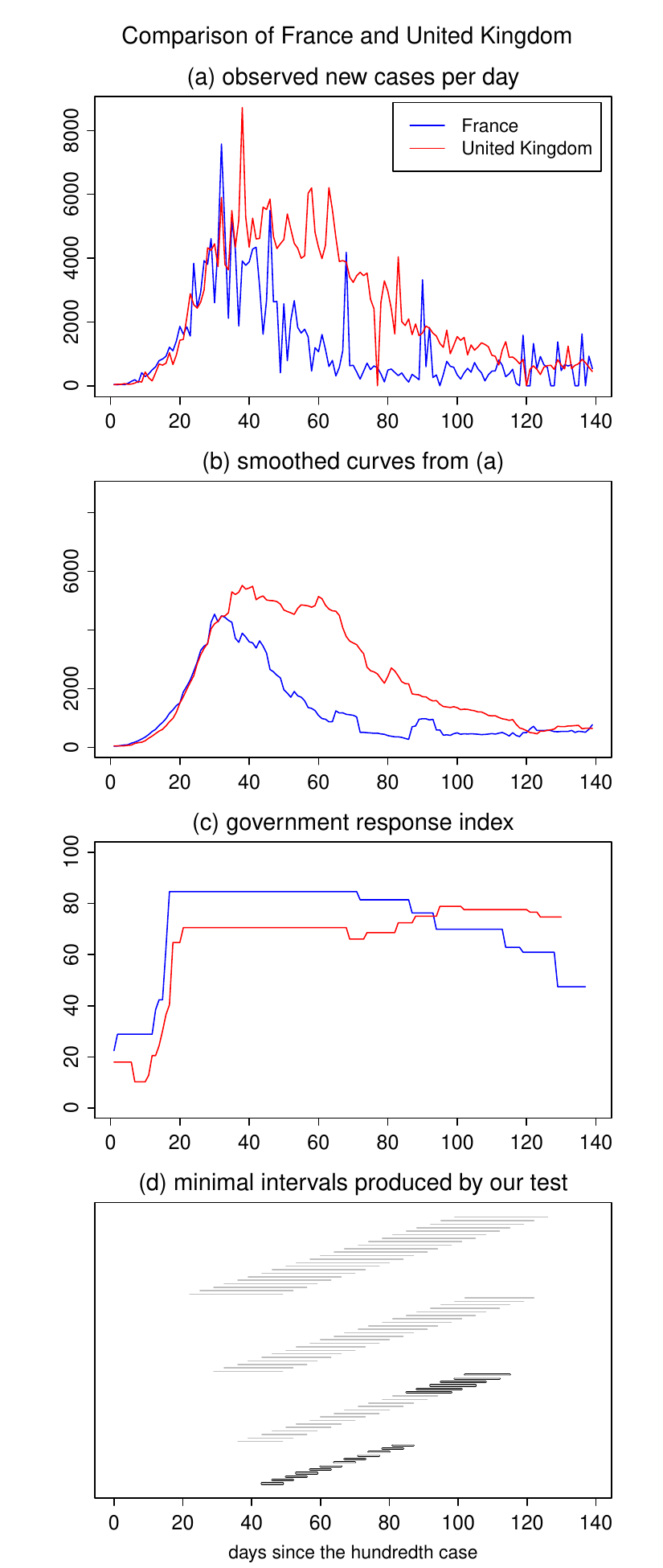}
\caption{Test results for the comparison of France and the UK.}
\end{minipage}
\end{figure}

\newpage
\subsection{Robustness checks for Section \ref{subsec:sim}}\label{s:subsec:robustness}

In what follows, we supplement the simulation experiments of Section \ref{subsec:sim} by some robustness checks. Specifically, we repeat the experiments with different values of the overdispersion parameter $\sigma$. The larger we choose $\sigma$, the more noise we put on top of the time trend, that is, on top of the underlying signal. Hence, by varying $\sigma$, we can assess how sensitive our test is to changes in the noise-to-signal ratio. We first repeat the size simulations for $\sigma = 10$ and $\sigma = 20$. The results are presented in Tables \ref{s:tab:sim:size:1} and \ref{s:tab:sim:size:2}, respectively. As can be seen, the empirical size numbers are very similar to those for $\sigma=15$ in Table \ref{tab:sim:size}. We next rerun the power simulations for $\sigma = 10$ and $\sigma = 20$, where we consider the two Scenarios A and B as in Section \ref{subsec:sim}. The results can be found in Tables \ref{s:tab:sim:power:1}--\ref{s:tab:sim:power:4}. They show that the test is much more powerful for $\sigma=10$ than for $\sigma=20$. This is what one would expect, since a higher value of $\sigma$ corresponds to a higher noise-to-signal ratio. In particular, the higher $\sigma$, the more noisy the data, and thus the more difficult it is to identify differences between the trend curves. Nevertheless, even in the very noisy case with $\sigma = 20$, our test has quite some power, which tends to increase swiftly as $T$ gets larger. 
\vspace{0.4cm}

\begin{table}[h!]
\footnotesize{
\caption{Empirical size of the test for $\sigma = 10$.}\label{s:tab:sim:size:1}
\newcolumntype{C}[1]{>{\hsize=#1\hsize\centering\arraybackslash}X}
\newcolumntype{Z}{>{\centering\arraybackslash}X}
\begin{tabularx}{\textwidth}{l Z@{\hskip 6pt}Z@{\hskip 6pt}Z Z@{\hskip 6pt}Z@{\hskip 6pt}Z Z@{\hskip 6pt}Z@{\hskip 6pt}Z} 
\toprule
 & \multicolumn{3}{c}{$n = 5$} & \multicolumn{3}{c}{$n = 10$} & \multicolumn{3}{c}{$n = 50$} \\
\cmidrule[0.4pt]{2-4} \cmidrule[0.4pt]{5-7} \cmidrule[0.4pt]{8-10}
 & \multicolumn{3}{c}{significance level $\alpha$} &\multicolumn{3}{c}{significance level $\alpha$} & \multicolumn{3}{c}{significance level $\alpha$} \\
 & 0.01 & 0.05 & 0.1  &  0.01 & 0.05 & 0.1  &  0.01 & 0.05 & 0.1 \\
\cmidrule[0.4pt]{1-10}
 $T = 100$ & 0.009 & 0.043 & 0.085 & 0.008 & 0.039 & 0.075 & 0.005 & 0.023 & 0.055 \\ 
  $T = 250$ & 0.011 & 0.047 & 0.095 & 0.010 & 0.050 & 0.094 & 0.009 & 0.039 & 0.079 \\ 
  $T = 500$ & 0.009 & 0.052 & 0.101 & 0.013 & 0.049 & 0.101 & 0.010 & 0.039 & 0.084 \\ 
\bottomrule
\end{tabularx}}
\end{table}

\begin{table}[h!]
\footnotesize{
\caption{Empirical size of the test for $\sigma = 20$.}\label{s:tab:sim:size:2}
\newcolumntype{C}[1]{>{\hsize=#1\hsize\centering\arraybackslash}X}
\newcolumntype{Z}{>{\centering\arraybackslash}X}
\begin{tabularx}{\textwidth}{l Z@{\hskip 6pt}Z@{\hskip 6pt}Z Z@{\hskip 6pt}Z@{\hskip 6pt}Z Z@{\hskip 6pt}Z@{\hskip 6pt}Z} 
\toprule
 & \multicolumn{3}{c}{$n = 5$} & \multicolumn{3}{c}{$n = 10$} & \multicolumn{3}{c}{$n = 50$} \\
\cmidrule[0.4pt]{2-4} \cmidrule[0.4pt]{5-7} \cmidrule[0.4pt]{8-10}
 & \multicolumn{3}{c}{significance level $\alpha$} &\multicolumn{3}{c}{significance level $\alpha$} & \multicolumn{3}{c}{significance level $\alpha$} \\
 & 0.01 & 0.05 & 0.1  &  0.01 & 0.05 & 0.1  &  0.01 & 0.05 & 0.1 \\
\cmidrule[0.4pt]{1-10}
$T = 100$ & 0.011 & 0.050 & 0.094 & 0.010 & 0.047 & 0.092 & 0.009 & 0.034 & 0.070 \\ 
  $T = 250$ & 0.009 & 0.047 & 0.088 & 0.008 & 0.044 & 0.085 & 0.006 & 0.032 & 0.062 \\ 
  $T = 500$ & 0.008 & 0.038 & 0.081 & 0.006 & 0.039 & 0.079 & 0.006 & 0.025 & 0.060 \\ 
\bottomrule
\end{tabularx}}
\end{table}

\pagebreak

\begin{table}[t!]
\footnotesize{
\caption{Power of the test in Scenario A for $\sigma = 10$.}\label{s:tab:sim:power:1}
\newcolumntype{C}[1]{>{\hsize=#1\hsize\centering\arraybackslash}X}
\newcolumntype{Z}{>{\centering\arraybackslash}X}
\begin{tabularx}{\textwidth}{l Z@{\hskip 6pt}Z@{\hskip 6pt}Z Z@{\hskip 6pt}Z@{\hskip 6pt}Z Z@{\hskip 6pt}Z@{\hskip 6pt}Z} 
\toprule
 & \multicolumn{3}{c}{$n = 5$} & \multicolumn{3}{c}{$n = 10$} & \multicolumn{3}{c}{$n = 50$} \\
\cmidrule[0.4pt]{2-4} \cmidrule[0.4pt]{5-7} \cmidrule[0.4pt]{8-10}
 & \multicolumn{3}{c}{significance level $\alpha$} &\multicolumn{3}{c}{significance level $\alpha$} & \multicolumn{3}{c}{significance level $\alpha$} \\
 & 0.01 & 0.05 & 0.1  &  0.01 & 0.05 & 0.1  &  0.01 & 0.05 & 0.1 \\
\cmidrule[0.4pt]{1-10}
$T = 100$ & 0.836 & 0.915 & 0.911 & 0.833 & 0.903 & 0.898 & 0.777 & 0.874 & 0.882 \\ 
  $T = 250$ & 0.986 & 0.971 & 0.938 & 0.984 & 0.956 & 0.918 & 0.980 & 0.961 & 0.924 \\ 
  $T = 500$ & 0.996 & 0.975 & 0.946 & 0.994 & 0.965 & 0.927 & 0.992 & 0.963 & 0.918 \\ 
\bottomrule
\end{tabularx}}
\vspace{0.5cm}
\end{table}

\begin{table}[t!]
\footnotesize{
\caption{Power of the test in Scenario A for $\sigma = 20$.}\label{s:tab:sim:power:2}
\newcolumntype{C}[1]{>{\hsize=#1\hsize\centering\arraybackslash}X}
\newcolumntype{Z}{>{\centering\arraybackslash}X}
\begin{tabularx}{\textwidth}{l Z@{\hskip 6pt}Z@{\hskip 6pt}Z Z@{\hskip 6pt}Z@{\hskip 6pt}Z Z@{\hskip 6pt}Z@{\hskip 6pt}Z} 
\toprule
 & \multicolumn{3}{c}{$n = 5$} & \multicolumn{3}{c}{$n = 10$} & \multicolumn{3}{c}{$n = 50$} \\
\cmidrule[0.4pt]{2-4} \cmidrule[0.4pt]{5-7} \cmidrule[0.4pt]{8-10}
 & \multicolumn{3}{c}{significance level $\alpha$} &\multicolumn{3}{c}{significance level $\alpha$} & \multicolumn{3}{c}{significance level $\alpha$} \\
 & 0.01 & 0.05 & 0.1  &  0.01 & 0.05 & 0.1  &  0.01 & 0.05 & 0.1 \\
\cmidrule[0.4pt]{1-10}
$T = 100$ & 0.144 & 0.275 & 0.352 & 0.115 & 0.231 & 0.304 & 0.048 & 0.120 & 0.163 \\ 
  $T = 250$ & 0.244 & 0.434 & 0.538 & 0.204 & 0.403 & 0.486 & 0.133 & 0.247 & 0.305 \\ 
  $T = 500$ & 0.296 & 0.563 & 0.662 & 0.273 & 0.511 & 0.603 & 0.175 & 0.338 & 0.433 \\ 
\bottomrule
\end{tabularx}}
\vspace{0.5cm}
\end{table}

\begin{table}[t!]
\footnotesize{
\caption{Power of the test in Scenario B for $\sigma = 10$.}\label{s:tab:sim:power:3}
\newcolumntype{C}[1]{>{\hsize=#1\hsize\centering\arraybackslash}X}
\newcolumntype{Z}{>{\centering\arraybackslash}X}
\begin{tabularx}{\textwidth}{l Z@{\hskip 6pt}Z@{\hskip 6pt}Z Z@{\hskip 6pt}Z@{\hskip 6pt}Z Z@{\hskip 6pt}Z@{\hskip 6pt}Z} 
\toprule
 & \multicolumn{3}{c}{$n = 5$} & \multicolumn{3}{c}{$n = 10$} & \multicolumn{3}{c}{$n = 50$} \\
\cmidrule[0.4pt]{2-4} \cmidrule[0.4pt]{5-7} \cmidrule[0.4pt]{8-10}
 & \multicolumn{3}{c}{significance level $\alpha$} &\multicolumn{3}{c}{significance level $\alpha$} & \multicolumn{3}{c}{significance level $\alpha$} \\
 & 0.01 & 0.05 & 0.1  &  0.01 & 0.05 & 0.1  &  0.01 & 0.05 & 0.1 \\
\cmidrule[0.4pt]{1-10}
$T = 100$ & 0.991 & 0.973 & 0.946 & 0.994 & 0.970 & 0.935 & 0.994 & 0.971 & 0.940 \\ 
  $T = 250$ & 0.993 & 0.969 & 0.941 & 0.993 & 0.959 & 0.919 & 0.991 & 0.960 & 0.925 \\ 
  $T = 500$ & 0.996 & 0.976 & 0.948 & 0.993 & 0.966 & 0.928 & 0.993 & 0.962 & 0.917 \\ 
\bottomrule
\end{tabularx}}
\vspace{0.5cm}
\end{table}

\begin{table}[t!]
\footnotesize{
\caption{Power of the test in Scenario B for $\sigma = 20$.}\label{s:tab:sim:power:4}
\newcolumntype{C}[1]{>{\hsize=#1\hsize\centering\arraybackslash}X}
\newcolumntype{Z}{>{\centering\arraybackslash}X}
\begin{tabularx}{\textwidth}{l Z@{\hskip 6pt}Z@{\hskip 6pt}Z Z@{\hskip 6pt}Z@{\hskip 6pt}Z Z@{\hskip 6pt}Z@{\hskip 6pt}Z} 
\toprule
 & \multicolumn{3}{c}{$n = 5$} & \multicolumn{3}{c}{$n = 10$} & \multicolumn{3}{c}{$n = 50$} \\
\cmidrule[0.4pt]{2-4} \cmidrule[0.4pt]{5-7} \cmidrule[0.4pt]{8-10}
 & \multicolumn{3}{c}{significance level $\alpha$} &\multicolumn{3}{c}{significance level $\alpha$} & \multicolumn{3}{c}{significance level $\alpha$} \\
 & 0.01 & 0.05 & 0.1  &  0.01 & 0.05 & 0.1  &  0.01 & 0.05 & 0.1 \\
\cmidrule[0.4pt]{1-10}
$T = 100$ & 0.438 & 0.636 & 0.704 & 0.404 & 0.598 & 0.669 & 0.277 & 0.449 & 0.526 \\ 
  $T = 250$ & 0.864 & 0.934 & 0.927 & 0.850 & 0.923 & 0.915 & 0.811 & 0.891 & 0.898 \\ 
  $T = 500$ & 0.960 & 0.968 & 0.949 & 0.961 & 0.964 & 0.935 & 0.945 & 0.961 & 0.941 \\ 
\bottomrule
\end{tabularx}}
\end{table}

\phantom{end}

\end{document}

%% file: macros.tex
\newcommand{\reals}{\mathbb{R}}

\newcommand{\naturals}{\mathbb{N}}

\newcommand{\pr}{\mathbb{P}}        
\newcommand{\ex}{\mathbb{E}}        
\newcommand{\var}{\textnormal{Var}} 

\newcommand{\ind}{\boldsymbol{1}} 

\newcommand{\X}{X}
\newcommand{\pairs}{\mathcal{S}}
\newcommand{\countries}{\mathcal{C}}
\newcommand{\intervals}{\mathcal{F}}
\newcommand{\indexset}{\mathcal{M}}

\theoremstyle{plain}

\newtheorem{theoremA}{Theorem}[section]
\newtheorem{propA}{Proposition}[section]
\newtheorem{corollaryA}{Corollary}[section]

\newtheorem{lemmaS}{Lemma}[section]

\makeatletter
\newcommand{\lefteqno}{\let\veqno\@@leqno}
\makeatother

\newcommand{\heading}[3]
{  \setcounter{page}{1}
   \begin{center}

   {\LARGE \textbf{#1}}
   \vspace{0.25cm}

   {\LARGE \textbf{#2}}
   \vspace{0.25cm}

   {\LARGE \textbf{#3}}
   \end{center}
}

\newcommand{\authors}[4]
{  
   \begin{center}
      \begin{minipage}[c][1.5cm][c]{5.5cm}
      \begin{center} 
      {\large #1} 
      \vspace{0.1cm}
      
      #2 
      \end{center}
      \end{minipage}
      \begin{minipage}[c][1.5cm][c]{5.5cm}
      \begin{center} 
      {\large #3}
      \vspace{0.1cm}
      
      #4 
      \end{center}
      \end{minipage}
   \end{center}
}
